\documentclass[11pt,leqno,twoside,a4paper]{article}
\usepackage[english]{babel}
\usepackage[utf8]{inputenc}
\usepackage{amsfonts,amssymb,amsmath,amsthm}
\usepackage{enumitem}
\usepackage[a4paper,left=30mm,right=30mm,top=25mm,bottom=30mm]{geometry}
\usepackage{graphicx} 
\usepackage{mathabx}
\usepackage{scalerel,stackengine}
\usepackage{verbatim}
\usepackage{calc}
\usepackage{color}
\usepackage{xcolor}
\usepackage{yfonts}
\usepackage{mathptmx}
\usepackage[T1]{fontenc}
\usepackage{lmodern}     
\usepackage[hidelinks]{hyperref}

\providecommand{\appendixname}{Appendix}

\usepackage{fancyhdr}
\fancypagestyle{plain}{%
	\fancyhf{}
	\fancyhead[LE,LO]{\thepage}

}

\newlength{\abstractleft}
\newlength{\abstractright}
\makeatletter
\renewenvironment{abstract}{%
	\par\small
	\list{}{%
		\setlength{\leftmargin}{\abstractleft}%
		\setlength{\rightmargin}{\abstractright}%
		\setlength{\listparindent}{0pt}%
		\setlength{\itemindent}{0pt}%
		\setlength{\parsep}{0pt}%
	}%
	\item\relax
	\noindent\textsc{Abstract.}\ %
	\ignorespaces
}{%
	\endlist
	\par\normalsize
}
\makeatother	



\newcommand{\ShortAuthors}{\scriptsize R. MAHADEVAN, F. OLIVARES, AND A. ZUNIGA} 
\newcommand{\ShortTitle}{\footnotesize Sharp three-particle fractional Hardy inequality} 

\fancypagestyle{plain}{%
	\fancyhf{}
	\fancyfoot[C]{\thepage}%
}

\usepackage{etoolbox}
\newif\ifinappendix
\pretocmd{\appendix}{%
	\inappendixtrue
	\addtocontents{toc}{\protect\inappendixtrue}%
}{}{}

\makeatletter

\renewcommand{\@seccntformat}[1]{%
	\ifinappendix
	\ifdefstring{#1}{section}{\appendixname~\thesection.\quad}{%
		\ifdefstring{#1}{subsection}{\appendixname~\thesubsection.\quad}{%
			\csname the#1\endcsname\quad
		}%
	}%
	\else
	\csname the#1\endcsname\quad
	\fi
}
\makeatletter

\newcommand{\TOCEntryFont}{}

\renewcommand{\l@section}[2]{%
	\begingroup
	\parskip=0pt
	\let\oldnumberline\numberline
	\ifinappendix
	\renewcommand{\numberline}[1]{\oldnumberline{\appendixname~##1.}}%
	\@dottedtocline{1}{1.5em}{6.8em}{{\TOCEntryFont #1}}{{\TOCEntryFont #2}}%
	\else
	\renewcommand{\numberline}[1]{\oldnumberline{##1.}}%
	\@dottedtocline{1}{1.5em}{3.0em}{{\TOCEntryFont #1}}{{\TOCEntryFont #2}}%
	\fi
	\endgroup
}

\renewcommand{\l@subsection}[2]{%
	\begingroup
	\parskip=0pt
	\let\oldnumberline\numberline
	\ifinappendix
	\renewcommand{\numberline}[1]{\oldnumberline{\appendixname~##1.}}%
	\@dottedtocline{2}{3.8em}{7.4em}{{\TOCEntryFont #1}}{{\TOCEntryFont #2}}%
	\else
	\renewcommand{\numberline}[1]{\oldnumberline{##1.}}%
	\@dottedtocline{2}{3.8em}{3.8em}{{\TOCEntryFont #1}}{{\TOCEntryFont #2}}%
	\fi
	\endgroup
}
\makeatother

\usepackage{titlesec}
\titleformat{\section}
{\normalfont\scshape\large\filcenter}
{\ifinappendix\appendixname~\thesection.\else\thesection.\fi}{0.8em}{}

\titleformat{\subsection}
{\normalfont\scshape\normalsize\filcenter}
{\ifinappendix\appendixname~\thesubsection.\else\thesubsection.\fi}{0.8em}{}

\titlespacing*{\section}{0pt}{1.5ex plus .2ex}{0.8ex}
\titlespacing*{\subsection}{0pt}{1.2ex plus .2ex}{0.6ex}

\titleformat{\subsubsection}
{\normalfont\scshape\normalsize}
{\thesubsubsection.}{0.8em}{}

\newtheorem{theorem}{Theorem}[section]
\newtheorem{lemma}[theorem]{Lemma}
\newtheorem{proposition}[theorem]{Proposition}
\newtheorem{corollary}[theorem]{Corollary}
\newtheorem{remark}{Remark}

\newcommand{\1}{\mathbf{1}}

\newcommand{\calB}{\mathcal{B}}

\newcommand{\calS}{\mathcal{S}}
\newcommand{\phihat}{\widehat{\varphi}}

\newcommand{\kernel}{K}
\newcommand{\valphan}{v_{\alpha,n}}
\newcommand{\om}[1]{\omega^{(#1)}_{\alpha,n}}
\newcommand{\off}[1]{\Omega_{d,#1}}
\newcommand{\uhat}{\widehat{u}}
\newcommand{\ucheck}{\widecheck{u}}
\newcommand{\x}{\mathbf{x}}

\newcommand{\CFH}{C_{fH}(d,s)}

\newcommand{\CMOZ}{C_{\rm MOZ}(d_1,d_2,d_3,d)}
\newcommand{\CLL}[1]{\mathbf{c}_{#1}}
\newcommand{\CGM}{C_{\rm GM}(\alpha_1,\alpha_2,\alpha_3,d)}
\newcommand{\CStein}{C_{\rm Stein}(\alpha,\beta,d)}
\newcommand{\C}{\mathbb{C}}
\newcommand{\F}{\mathcal{F}}
\newcommand{\GM}{\mathcal{B}^{\rm GM}}
\newcommand{\N}{\mathbf{N}}
\newcommand{\Nn}{\mathcal{N}_s}
\newcommand{\R}{\mathbf{R}}
\newcommand{\rem}{\mathsf{R}_s}
\newcommand{\remR}{\textgoth{R}_{N,s}}
\renewcommand{\S}{\mathbf{S}}
\newcommand{\V}{V_{s,3}}
\newcommand{\optC}{C(s,d,3)}

\newcommand{\sLap}{(-\Delta)^s}
\newcommand{\sLapj}[1]{(-\Delta_{#1})^s}

\newcommand{\intd}{\int_{\R^d}}

\newcommand{\innd}[2]{\left\langle #1,#2 \right\rangle_{L^2(\R^d)}}
\newcommand{\inn}[2]{\left\langle #1,#2 \right\rangle}
\renewcommand{\epsilon}{\ensuremath{\varepsilon}}

\DeclareMathOperator*{\esssup}{ess\,sup\,}
\DeclareMathOperator{\supp}{supp}

\colorlet{darkred}{red!80!black}
\colorlet{darkblue}{blue!80!black}
\colorlet{darkgreen}{green!60!black}

\stackMath
\newcommand\wwidehat[1]{%
	\savestack{\tmpbox}{\stretchto{%
			\scaleto{%
				\scalerel*[\widthof{\ensuremath{#1}}]{\kern.1pt\mathchar"0362\kern.1pt}%
				{\rule{0ex}{\textheight}}
			}{\textheight}%
		}{2.4ex}}%
	\stackon[-6.9pt]{#1}{\tmpbox}%
}
\stackMath

\newcommand\wwidecheck[1]{%
	\savestack{\tmpbox}{\stretchto{%
			\scaleto{%
				\scalerel*[\widthof{\ensuremath{#1}}]{%
					\kern.1pt
					\raisebox{0.25ex}{%
						\scalebox{0.9}[1.5]{\rotatebox[origin=c]{180}{$\mathchar"0362$}}%
					}%
					\kern.1pt
				}{\rule{0ex}{\textheight}}%
			}{\textheight}%
		}{2.4ex}}%
	\stackon[.2pt]{#1}{\tmpbox}%
}

\providecommand{\keywords}[1]
{
	\small
	\textit{Keywords: #1}
}

\title{\large\bfseries\MakeUppercase{A sharp three-particle fractional Hardy inequality and an angular Selberg-type identity}$^{\dagger}$} 
\author{\small\MakeUppercase{Rajesh Mahadevan, \; Franco Olivares,\; and\; Andres Zuniga}}
\date{}

\newcommand{\addressblock}[1]{%
	\par\noindent 
	{\scshape #1}\par\medskip
}

\newcommand{\EmailB}[1]{%
	\textit{E-mail address}: %
	{\upshape\ttfamily\bfseries\href{mailto:#1}{#1}}%
}

\makeatletter
\renewcommand{\tableofcontents}{%
	\@starttoc{toc}%
}
\makeatother
\begin{document}
	\maketitle
	\thispagestyle{plain}
	\vspace{-2.25em} 
	
	\begingroup
	\renewcommand{\thefootnote}{\fnsymbol{footnote}}
	\setcounter{footnote}{2}
	\footnotetext{%
		\textit{Date:} 28 May 2026.\par 
		\par
		\par
		\textcopyright\ 2026 by the authors. Faithful reproduction of this article, in its entirety,
		by any means is permitted for non-commercial purposes.%
	}
	\endgroup
	
	\setcounter{footnote}{0}
	
	\begin{abstract}
		We establish a sharp three-particle fractional Hardy inequality for the Laplacian of order $s\in(0,1)$ in dimension $d\ge 4-2s$ (Theorem~\ref{thm:main:1}), involving an explicit intrinsically three-body interaction potential $\V$. The inequality holds with the optimal two-particle fractional Hardy constant $\CFH$, which is shown to be sharp relative to the fixed potential $\V$. This potential $\V$ strictly dominates the standard pairwise Coulomb-type interaction and captures genuine three-body effects. As a consequence, we derive a nontrivial many-particle fractional Hardy inequality for $N\geq 3$, and, in the regime $N>d+1$, obtain an improved Coulomb-type inequality with a strictly larger constant, agreeing in spirit with the results of Hoffmann-Ostenhof et al.~\cite{hoffmann2008many} and Lundholm~\cite{lundholm2015geometric}. The proof relies on a fractional \emph{ground-state representation} method adapted to three-particle interactions, combined with an explicit evaluation of the resulting nonlocal interaction term. This evaluation is achieved through a \emph{new singular integral identity} of Selberg type (Theorem~\ref{thm:main:2}), extending the three-fold formula of Grafakos-Morpurgo~\cite{grafakos1999selberg} beyond the radial setting. This identity provides the analytic mechanism underlying the emergence of the three-body potential and may be of independent interest in harmonic analysis.
	\end{abstract}
	
	\begin{center}
		{\large\textsc{Contents}}
	\end{center}
	\vspace{-1.0em}
	\tableofcontents
	
	\section{Introduction and main results}
	
	Hardy-type inequalities play a central role in analysis and mathematical physics, particularly in regimes where singular potentials compete with kinetic energy. In the two-particle setting, the classical Hardy inequality~\cite{hardy1952inequalities} asserts that, for $d\geq 3$,
	\begin{equation}\label{eq:intro:hardyIneq}
		\intd |\nabla u(x)|^2\,dx
		\ \geq\
		\frac{(d-2)^2}{4}\intd \frac{|u(x)|^2}{|x|^2}\,dx\qquad\text{for all } u\in C^\infty_c(\R^d),
	\end{equation}
	while its nonlocal counterpart, the fractional Hardy inequality, states that for $d\ge 2$ and $s\in(0,1)$,
	\begin{equation}\label{eq:intro:fH}
		\langle (-\Delta)^s u,u\rangle_{L^2(\R^d)}
		\ \geq\
		\CFH \intd \frac{|u(x)|^2}{|x|^{2s}}\,dx
		\qquad \text{for all } u\in C^\infty_c(\R^d),
	\end{equation}
	where 
	\begin{equation}\label{eq:def:fracHardy:constant}
		\CFH
		\ =\
		2^{2s}\frac{\Gamma^2(\frac{d+2s}{4})}{\Gamma^2(\frac{d-2s}{4})}
	\end{equation}
	is the sharp constant in~\eqref{eq:intro:fH}. We refer to Herbst~\cite{herbst1977spectral}, Kato~\cite{kato1966perturbation}, Yafaev~\cite{yafaev1999sharp}, and Frank--Lieb--Seiringer~\cite{frank2008hardy} for background, sharp constants, and further developments in the fractional setting. In the special case $s=\frac12$, this reduces to Kato's inequality and is closely related to the massless relativistic kinetic energy.
	
	We now turn to the many-particle setting. Let
	\begin{equation}\label{eq:intro:def:omega}
		\off{N}
		\ :=\
		\{x=(x_1,\dots,x_N)\in \R^{dN}: \;x_i\neq x_j \;\text{ for }\; i\neq j\}
	\end{equation}
	denote the collision-free configuration space. As a direct consequence of the two-particle inequalities~\eqref{eq:intro:hardyIneq} and~\eqref{eq:intro:fH}, one obtains the standard many-particle inequalities with Coulomb-type interactions:
	\begin{equation}\label{eq:intro:hardyQuadraticVersion}
		\int_{\R^{dN}} |\nabla u(x)|^2\,dx
		\ \geq\
		\frac{(d-2)^2}{2(N-1)}\int_{\R^{dN}}
		\left(\sum_{1\le i<j\le N}\frac{1}{|x_i-x_j|^2}\right)|u(x)|^2\,dx,
	\end{equation}
	and
	\begin{equation}\label{eq:intro:fracHardyQuadraticVersion} 
		\int_{\R^{dN}} u(x)\left(\sum_{i=1}^N \sLapj{i}u(x)\right)\,dx
		\ \geq\
		\frac{2\CFH}{N-1}\int_{\R^{dN}}\left(\sum_{1\le i<j\le N}\frac{1}{|x_i-x_j|^{2s}}\right)|u(x)|^2\,dx,
	\end{equation}
	valid for $u\in C^\infty_c(\off{N})$. In particular, when $N=3$, the constants coincide with the two-particle values, reflecting the absence of genuine three-body effects. 
	
	In the local case, Hoffmann-Ostenhof \emph{et al.}~\cite[Theorem 2.1]{hoffmann2008many} showed that the bound above,~\eqref{eq:intro:hardyQuadraticVersion}, can be improved by incorporating three-particle interactions:
	\begin{align}
		\int_{\R^{Nd}}|\nabla u(x)|^2\,dx \ \geq\ C_{mpH}(d,N)\int_{\R^{Nd}}\left(\sum_{1\leq i<j\leq N}\frac{1}{|x_i-x_j|^2}\right)|u(x)|^2\,dx,
	\end{align}
	for all $u\in C^{\infty}_c(\off{N})$, where the constant $C_{mpH}(d,N)$ satisfies
	\begin{align}
		C_{mpH}(d,N) \ \geq\ (d-2)^2\max\left\{\frac{1}{N},\frac{1}{1+\sqrt{1+\frac{3(d-2)^2(N-2)(N-1)}{2(d-1)^2}}}\right\}.
	\end{align}
	This improvement is driven by genuine three-body interactions and was further refined by Lundholm~\cite{lundholm2015geometric}. In particular, for $N=3$, the resulting constant is at least $4/3$ times larger than that in the purely pairwise inequality~\eqref{eq:intro:hardyQuadraticVersion}.
	
	The main objective of the present work is to establish a nonlocal analogue of this phenomenon. More precisely, we first obtain a sharp three-particle fractional Hardy inequality exhibiting an intrinsically three-body potential, and then derive corresponding many-particle consequences. Our approach to the theorem is intrinsically nonlocal.
	\begin{theorem}\label{thm:main:1}
		Let $s\in(0,1)$ and let $d\ge 4-2s$. Let $\off{3}$ be the collision-free set in $\R^{3d}$ defined in~\eqref{eq:intro:def:omega}. Then, for every $u\in C^\infty_c(\off{3})$,
		\begin{equation}\label{eq:thm:1:inequality}
			\int_{\R^{3d}} u(x)\left(\sum_{i=1}^3 \sLapj{i}u(x)\right)\,dx
			\ \geq\
			\CFH\int_{\R^{3d}} \V(x)\,|u(x)|^2\,dx,
		\end{equation}
		where $\CFH$ is the sharp constant in the fractional Hardy inequality~\eqref{eq:def:fracHardy:constant}, and this constant is also sharp for~\eqref{eq:thm:1:inequality}. The three-particle Hardy potential is given by
		\begin{equation}\label{eq:thm:1:def:potential}
			\V(x_1,x_2,x_3)
			\ :=\
			\sum_{{\rm cyclic}}
			\frac{|x_k-x_j|^{2s}}
			{|x_k-x_i|^{2s}\,|x_j-x_i|^{2s}},
			\qquad (x_1,x_2,x_3)\in \off{3},
		\end{equation}
		where ``{\it \mbox{\rm cyclic}}'' denotes the cyclic sum over the indices $(i,j,k)$ in the set
		\[
		\{(1,2,3),\quad (2,3,1),\quad (3,1,2)\}. 
		\]
	\end{theorem}
	
	We now emphasise the main structural features of Theorem~\ref{thm:main:1}. First, the potential $\V$ is genuinely of three-body type, in that it couples all relative distances in a nontrivial manner. Second, for every $s\in(0,1)$ and every $(x_1,x_2,x_3)\in\off{3}$, one has
	\begin{equation}\label{eq:intro:potential-bound}
		\V(x_1,x_2,x_3)
		\ \geq\
		\sum_{1\le i<j\le 3}\frac{1}{|x_i-x_j|^{2s}}\;,
	\end{equation}
	as a direct consequence of the elementary inequality
	\begin{equation}\label{eq:lem:h-o:1}
		\frac{c^{2s}}{a^{2s}b^{2s}} + \frac{b^{2s}}{a^{2s}c^{2s}} + \frac{a^{2s}}{b^{2s}c^{2s}}
		\ \geq\
		\frac{1}{a^{2s}}+\frac{1}{b^{2s}}+\frac{1}{c^{2s}}.
	\end{equation}
	Equality holds if and only if the three particles form an equilateral configuration. In particular, the potential $\V$ strictly dominates the standard pairwise Coulomb-type interaction, and the difference is encoded in a non-negative defect term $\rem(x_1,x_2,x_3)$, which vanishes precisely in the equilateral case. Third, although the coefficient $\CFH$ coincides with that appearing in the pairwise inequality~\eqref{eq:intro:fracHardyQuadraticVersion} when $N=3$, it is sharp relative to the three-body potential $\V$, but not necessarily with respect to the purely pairwise interaction.
	
	Moreover, Theorem~\ref{thm:main:1} yields a refined $N$-particle fractional Hardy inequality featuring an explicit non-negative defect term. When $N>d+1$, this defect term itself controls the pairwise Coulomb interaction, leading to a strict improvement of the constant $\CFH$, agreeing in spirit with the local results of Hoffmann-Ostenhof \emph{et al.}~\cite{hoffmann2008many} and Lundholm~\cite{lundholm2015geometric}. The threshold $N>d+1$ is optimal.
	
	The second principal contribution of the paper is the identification of the analytic mechanism underlying the emergence of the three-body interaction. In the fractional ground-state representation, the interaction term generated by the Leibniz rule can still be computed explicitly, but only through a new singular integral identity involving both radial and angular components. Theorem~\ref{thm:main:1} ultimately relies on this second main result, namely, a Selberg-type singular integral identity.

	\begin{theorem}[\textbf{Angular Selberg-type correlation identity}]\label{thm:main:2}
		Let $d\geq 2$, let $x,y,z\in\R^d$ be distinct points and let $0<d_1,d_2,d_3<d$ be exponents such that $d_1+d_2+d_3=2d$. Then
		\begin{align}
			&\intd
			\frac{(x-t)\cdot(y-t)}{|x-t|^{d_2+1}|y-t|^{d_3+1}}
			\frac{1}{|z-t|^{d_1}}\,dt\notag\\
			&\hspace{5em}
			=\ \CMOZ\;
			\frac{(x-z)}{|x-z|^{d-d_3+1}}
			\cdot
			\frac{(y-z)}{|y-z|^{d-d_2+1}}\,
			\frac{1}{|x-y|^{d-d_1}},\label{eq:thm:MOZ:1}
		\end{align}
		where the constant $\CMOZ$ is explicitly given by
		\begin{equation}\label{eq:thm:MOZ:2}
			\CMOZ
			\ =\ 
			\pi^{\frac d2}
			\frac{\Gamma(\frac{d-d_1}{2})}{\Gamma(\frac{d_1}{2})}
			\frac{\Gamma(\frac{d-d_2+1}{2})}{\Gamma(\frac{d_2+1}{2})}
			\frac{\Gamma(\frac{d-d_3+1}{2})}{\Gamma(\frac{d_3+1}{2})}.
		\end{equation}
	\end{theorem}
	
	\noindent	
	From the viewpoint of the present paper, Theorem~\ref{thm:main:2} is precisely the mechanism that converts the nonlocal interaction term arising in the ground-state representation into the explicit three-particle potential $\V$. More broadly, it extends the three-fold Selberg formula of Grafakos-Morpurgo~\cite{grafakos1999selberg} to a kernel with degree-one angular structure, while still admitting an explicit closed-form evaluation under the homogeneity condition $d_1+d_2+d_3=2d$. In the proof of Theorem~\ref{thm:main:1}, this identity provides the explicit evaluation of the collision term produced by the fractional Leibniz rule, and the subsequent application of the law of cosines yields the three-body potential appearing there. Beyond its role in the present argument, Theorem~\ref{thm:main:2} shows that the Grafakos-Morpurgo formula is not confined to purely radial kernels, and may therefore be of independent interest in harmonic analysis.\smallskip
	
	The main theorem immediately yields an $N$-particle consequence by summing over triples. In this sense, the many-particle result is a corollary of the three-particle analysis and, therefore, cannot be regarded as the primary result of this paper. Consequently, as in Hoffmann-Ostenhof \emph{et al.}~\cite{hoffmann2008many} for the local case, the key to the many-particle fractional Hardy inequality is the three-particle fractional Hardy inequality.
	
	\begin{corollary}\label{cor:many-particle}
		Let $s\in(0,1)$ and let $d\ge 4-2s$. Let $\off{N}\subset\R^{dN}$ be the collision-free set defined in~\eqref{eq:intro:def:omega}. For any $N\ge 3$, define
		\[
		V_{N,3}(x_1,\ldots,x_N)
		\ :=\ 
		\sum_{1\le i<j<k\le N}\V(x_i,x_j,x_k)
		\qquad \text{for } (x_1,\ldots,x_N)\in\off{N},
		\]
		where $\V$ is the three-particle Hardy potential in~\eqref{eq:thm:1:def:potential}. Then, for every $u\in C^\infty_c(\off{N})$,
		\begin{align}
			&\int_{\R^{dN}} u(x)\left(\sum_{i=1}^N \sLapj{i}u(x)\right)\,dx\notag\\
			&\hspace{4em}
			\ \geq\ 
			\frac{2\,\CFH}{(N-1)(N-2)}
			\int_{\R^{dN}}V_{N,3}(x_1,\ldots,x_N)|u(x)|^2\,dx.
			\label{eq:cor:many-particle:Hardy}
		\end{align}
	\end{corollary}
	
	Let us record two structural limitations of our approach. First, although in Theorem~\ref{thm:main:1} the coefficient $\CFH$ is sharp for the potential $\V$, passing to the limit as $s\uparrow 1$ does not recover the sharp constant in the local case. This is further explained in Remark~\ref{rem:behavior:s}. Second, the dimensional restriction $d\geq 4-2s$ arises naturally from the singular structure of the correlation term in the proof (see Remark~\ref{rem:dim:restriction} for further explanation). Since $s \in (0,1)$, the condition $d\geq 4-2s$ implies $d \ge 3$. In the physically relevant case $d=3$, this condition requires $s\geq \tfrac{1}{2}$, and this includes the ultra-relativistic case $(d,s)=(3,\tfrac{1}{2})$. We stress, however, that the condition $d\geq 4-2s$ should presently be viewed as a restriction of the proof and not as an obstruction to the validity of Theorem~\ref{thm:main:1} itself in the remaining cases. What our analysis shows is that, in the interaction computation underlying Lemma~\ref{lem:limit:interaction}, the boundary terms produced by the Green-type reduction cease to be integrable, or at least to be controllable by the present argument, once $d+2s<4$. The authors do not presently know of any counterexample showing that the inequality fails below this threshold. In this sense, the range $d\geq 4-2s$ marks the boundary of applicability of our ground-state-representation method combined with the explicit correlation identity. It remains an open problem to determine whether the theorem, or perhaps a suitable adaptation of it, continues to hold in the complementary regime.

	The proof strategy may be summarised as follows. We begin with a fractional ground-state representation adapted to three particles. The natural ground state is not directly admissible, and we therefore introduce a truncated approximation. To estimate the kinetic energy, we employ a fractional Leibniz rule. This reduces the proof of Theorem~\ref{thm:main:1} to the explicit evaluation of a nonlocal interaction term. The latter is handled by Theorem~\ref{thm:main:2}, whose proof proceeds by a Fourier-transform analysis of a new singular correlation kernel in the setting of tempered distributions. The resulting identity yields the explicit three-body potential $\V$ and hence the sharp inequality. 
	
	The remainder of the paper is organised as follows. In Section~\ref{sec:prelim} we collect the notation and analytic tools used throughout the paper, including the fractional ground-state estimate and the fractional Leibniz rule. In Section~\ref{sec:proof:main_thm} we prove Theorem~\ref{thm:main:1}, establish sharpness of the coefficient $\CFH$, and derive refined geometric and many-particle consequences of the three-body potential, including a strict quantitative improvement of the pairwise $N$-particle bound when $N>d+1$. In Section~\ref{subsec:new_integral_formula}, we prove the angular Selberg-type identity, Theorem~\ref{thm:main:2}. The appendix contains auxiliary technical results used in the proofs of the main theorems. 
	
	\section{Preliminaries}\label{sec:prelim}
	
	In this section we collect the notation and analytic tools used throughout the paper. We begin with basic facts on tempered distributions and Fourier transforms, and then recall the singular integral identities and properties of the fractional Laplacian needed in the proofs of Theorem~\ref{thm:main:1} and~Theorem~\ref{thm:main:2}.\medskip
	
	\subsection{Distributions and Fourier transforms}
	We use standard notation for Schwartz functions and tempered distributions. It is customary to denote by $\calS(\R^d)$ the Schwartz space of rapidly decaying smooth functions and by $\calS'(\R^d)$ its topological dual. 
	
	For $a\in\R^d$ and $\rho\in\R\setminus\{0\}$, we define the translation, dilation and reflection of $\varphi\in\calS(\R^d)$ by 
	\begin{equation}\label{eq:def:trans-dil-refl-function}
		(\tau^{a} \varphi)(x)\ :=\  \varphi(x-a)\,, \quad   (\delta^{\rho}\varphi)(x)\ :=\ \varphi(\rho x)\,, \quad\widetilde{\varphi}(x)\ :=\ \varphi(-x) \quad \mbox{ for } x \in \R^d\,.
	\end{equation}
	The corresponding actions on tempered distributions $T\in\calS'(\R^d)$ are defined by duality:
	\begin{equation}\label{eq:def:trans-dil-refl-dist}
		\langle \tau^{a}T \,, \,\varphi \rangle\ =\ \langle T \,, \,\tau^{-a}\varphi \rangle \,, \quad
		\langle \delta^{\rho}T \,, \, \varphi \rangle\ =\ \langle T \,, \, |\rho|^{-d} (\delta^{1/\rho}\varphi)\rangle\, , \quad
		\langle \widetilde{T} \,, \,\varphi \rangle\ =\ \langle T \,, \, \widetilde{\varphi} \rangle \,. 
	\end{equation}
	If $\psi\in\calS(\R^d)$ and $T\in\calS'(\R^d)$, then the product $\psi T$ and the convolution $\psi\ast T$ are defined by duality:
	\begin{equation}\label{eq:prelim:def:Fouriertempered-mult}
		\langle \psi T \,, \,\varphi \rangle \ =\ 
		\langle T \,, \, \psi \varphi \rangle,
	\end{equation} 	
	and 
	\begin{equation}\label{eq:prelim:def:Fouriertempered-conv}
		\langle  \psi \ast T \,, \,\varphi \rangle \ =\ 
		\langle T \,, \, \widetilde{\psi} \ast \varphi \rangle, 
	\end{equation} 
	for every $\varphi \in \mathcal{S}(\R^d)$.
	
	Throughout the paper we adopt the Fourier transform convention 
	\begin{equation}\label{eq:prelim:def:Fourier}
		\uhat(\xi)\ :=\ \intd e^{-2\pi i\xi\cdot x}u(x)\,dx,
	\end{equation}
	for every $ u\in\mathcal{S}(\R^d)$; see Stein~\cite{stein1971fourier}, Lieb-Loss~\cite{lieb2001analysis} or Grafakos~\cite{grafakos2008classical}. With this definition, the inverse Fourier transform is given by 
	\begin{equation}\label{eq:prelim:def:Fourierinverse}
		\ucheck(x)\ :=\ \intd e^{2\pi ix\cdot \xi}u(\xi)\,d\xi,
	\end{equation} 
	for every $u\in\mathcal{S}(\R^d)$. 
	Under this convention one has, for Schwartz functions $f$ and $g$,
	\begin{equation}\label{FTprop-functions}
		\widehat{\widetilde{f}}\ =\  \widetilde{\widehat{f}}\,, \qquad \widehat{\widehat{f}}\ =\ \widetilde{f}\,, \qquad \widehat{\tau^yf}(\xi)\ =\ e^{-2\pi i y \cdot \xi}\widehat{f}(\xi)\,, \qquad 
		\widehat{f\ast g}(\xi)\ =\ \widehat{f}(\xi) \widehat{g}(\xi)\,;
	\end{equation}
	see Grafakos~\cite[Proposition 2.2.11]{grafakos2008classical}.
	\smallskip
	
	\noindent For $T\in \calS'(\R^d)$, the Fourier transform is defined by duality:
	\begin{equation}\label{eq:prelim:def:Fouriertempered}
		\langle \widehat{T} \,, \,\varphi \rangle \ =\
		\langle T \,, \, \widehat{\varphi} \rangle,
	\end{equation} 
	for all $\varphi\in\mathcal{S}(\R^d)$. We shall use the following standard identities: for $\psi\in\calS(\R^d)$ and $T\in\calS'(\R^d)$,
	\begin{equation}\label{FTprop-distributions}
		\widehat{\widetilde{T}} \ =\ \widetilde{\widehat{T}}\,, \qquad \widehat{\widehat{T}}\ =\ \widetilde{T}\,, \qquad \widehat{\tau^yT}\ =\ e^{-2\pi i y \cdot \xi }\,\widehat{T}\,, \qquad \widehat{\psi T}\ =\ \widehat{\psi} \ast \widehat{T}\,\,\;\text{ in }\calS'(\R^d).
	\end{equation}
	See, for instance, Grafakos~\cite[Proposition 2.3.22]{grafakos2008classical}.	
	\smallskip 
	
	A key role in the sequel is played by the Riesz kernels. For $0 < \alpha < d$, we define
	\begin{equation}\label{eq:def:Rieszkernel}
		K_\alpha (x)\ :=\ 
		\frac{1}{\gamma(\alpha)}|x|^{-(d-\alpha)} ,\qquad \gamma(\alpha)\ :=\ 2^{\alpha}\pi^{\frac d2}  \dfrac{\Gamma(\frac{\alpha}{2})}{\Gamma(\frac{d-\alpha}{2})}\,,
	\end{equation}
	following Stein~\cite{stein1986singular}. Since $K_\alpha$ is locally integrable and bounded at infinity, it defines a tempered distribution. Moreover, for $0<\alpha<d$ with $\alpha\neq d-2$,  
	\begin{equation}\label{lap-Riesz}
		\Delta(|x|^{-\alpha})\ =\ (-\alpha)(d-\alpha-2)|x|^{-\alpha-2}\,.
	\end{equation}
	We recall the following standard identities for Riesz kernels and their Fourier transforms, which will be used repeatedly in the sequel. We use the classical Fourier transform formula, understood in the sense of tempered distributions:
	\begin{equation}\label{eq:Fourier:Lieb-Loss:aux:1}
		\wwidehat{|\cdot|^{-(d-\alpha_1)}} \ =\
		\frac{\CLL{\alpha_1}}{\CLL{d-\alpha_1}}\,|\cdot|^{-\alpha_1},
	\end{equation}
	where 
	\begin{equation}\label{eq:Fourier:Lieb-Loss:constant}
		\CLL{\mu}\ :=\
		\pi^{-\frac{\mu}{2}}\Gamma(\tfrac{\mu}{2}),\quad \text{ for }\mu>0\,;
	\end{equation}
	see Lieb-Loss~\cite[Theorem~5.9]{lieb2001analysis} or Stein \cite[Chapter 5:~Lemma 2]{stein1986singular}. One can readily check that the constants $\gamma(\alpha)$ and $c_\alpha$ are related by
	\begin{equation}\label{gamma-c}
		(2 \pi)^{-\alpha} \gamma(\alpha)\ =\ \dfrac{\CLL{\alpha}}{\CLL{d-\alpha}}.
	\end{equation}
	Another identity we shall also use is
	\begin{equation}\label{lem:app:Fourier}
		\wwidehat{\dfrac{x_j}{|x|^{d-\alpha+1}}}(\xi)\ =\
		\frac{\CLL{\alpha+1}}{i\,\CLL{d-\alpha+1}}\,
		\frac{\xi_j}{|\xi|^{\alpha+1}},\quad \text{ for }\;j\in\{1,\ldots,d\}.
	\end{equation}
	More generally, if $P_k\in\C[x]$ is a homogeneous harmonic polynomial of degree $k$ and $\sigma\in\C$ has $0<\operatorname{Re}\sigma<d$, then
	\begin{equation}\label{eq:Stein Weiss th 4.1}
		\wwidehat{\,P_k(x)\,|x|^{-d-k+\sigma}\,}(\xi)
		\ = \ \frac{\CLL{k+\sigma}}{i^k\,\CLL{d+k-\sigma}}\,
		P_k(\xi)\,|\xi|^{-k-\sigma}.
	\end{equation}
	See Stein-Weiss~\cite[Theorem 4.1]{stein1971fourier} for this general formula.
	\medskip
	
	\subsection{Singular integral formulae}
	The Riesz potential operator $I_\alpha$ is defined by convolution with the Riesz kernel $K_{\alpha}$: 
	\begin{equation}\label{eq:def:RieszOperator}
		I_\alpha(\varphi)(x)\ :=\
		\frac{1}{\gamma(\alpha)}\intd \frac{1}{|x-y|^{d-\alpha}}\,\varphi(y)\, dy \quad \mbox{ for all } \varphi \in \calS(\R^d)\,,
	\end{equation}
	where $\gamma(\alpha)$ is given in~\eqref{eq:def:Rieszkernel}. 
	
	Equivalently, using \eqref{eq:Fourier:Lieb-Loss:aux:1} and \eqref{gamma-c}, it can be defined in terms of its Fourier transform:
	\begin{equation}\label{FT-Riesztr}
		\wwidehat{I_\alpha  (\varphi)}(\xi)\ =\
		(2\pi)^{-\alpha}|\xi|^{-\alpha} \widehat{\varphi}(\xi).
	\end{equation}
	The Riesz potentials satisfy the semigroup property: for $0<\alpha,\beta<d$,    
	\begin{equation}\label{semigp-prop}		
		I_\alpha \circ I_\beta \ =\  I_{\alpha+\beta}, \quad \mbox{ whenever }\;\alpha + \beta < d;
	\end{equation}
	see Stein~\cite[Chapter 5]{stein1986singular}. This can be deduced in several ways, for instance, as a consequence of \eqref{FT-Riesztr}. The property can be written as an integral formula for the convolution of singular potentials as follows: given distinct $x,y\in\R^d$ ($d\geq 1$) and exponents $0<\alpha,\beta<d$ satisfying $\alpha+\beta>d$, there holds (see Stein \cite[Chapter 5:~equation~(8)]{stein1986singular})
	\begin{equation}\label{eq:prelim:stein:formula}
		\int_{\R^d}|x-t|^{-\alpha}|y-t|^{-\beta}\, dt\ =\
		\CStein\,|x-y|^{d-\alpha-\beta},
	\end{equation}
	where 
	\[
	\CStein \ :=\
	\dfrac{\gamma(d-\alpha)\gamma(d-\beta)}{\gamma(2d-(\alpha+\beta))}\ =\
	\pi^{\frac{d}{2}}\frac{\Gamma(\frac{d-\alpha}{2})\Gamma(\frac{d-\beta}{2})\Gamma(\frac{\alpha+\beta-d}{2})}{\Gamma(\frac{\alpha}{2})\Gamma(\frac{\beta}{2})\Gamma(d-\frac{\alpha+\beta}{2})} \,.
	\]
	The conditions $\alpha, \beta < d$ and $\alpha+\beta >d$ are sufficient for the integrability of the integrand on the left-hand side of \eqref{eq:prelim:stein:formula} and the resulting integral is homogeneous of degree $d-\alpha-\beta$ in $|x-y|$. As we shall find it useful, we introduce the notation
	\begin{equation}\label{beta-integral}
		\calB_{\alpha,\beta}(x,y)\ :=\
		\int_{\R^d}|x-t|^{-\alpha}|y-t|^{-\beta}\, dt\,,
	\end{equation}	
	and suppress the subscripts when no confusion can arise. 
	
	Formula~\eqref{eq:prelim:stein:formula} admits the following three-fold analogue due to Grafakos and Morpurgo~\cite[Theorem~1]{grafakos1999selberg}. Given $d\geq 3$, $x,y,z\in\R^d$ distinct, and exponents $0<\alpha_1,\alpha_2,\alpha_3<d$ with $\alpha_1+\alpha_2+\alpha_3=2d$, there holds
	\begin{align}\label{eq:prelim:grafakos:formula}
		& \int_{\R^d}|z-t|^{-\alpha_1}|x-t|^{-\alpha_2}|y-t|^{-\alpha_3}\,dt\notag\\
		&\qquad\ =\ \CGM\,|x-y|^{\alpha_1-d}|y-z|^{\alpha_2-d}|z-x|^{\alpha_3-d}\,,
	\end{align}
	with 
	\[
	\CGM \ :=\
	\pi^{\frac{d}{2}}\frac{\Gamma(\frac{d-\alpha_1}{2})\Gamma(\frac{d-\alpha_2}{2})\Gamma(\frac{d-\alpha_3}{2})}{\Gamma(\frac{\alpha_1}{2})\Gamma(\frac{\alpha_2}{2})\Gamma(\frac{\alpha_3}{2})}\,.
	\]
	Homogeneity with respect to $(x,y,z)$ on either side forces the constraint $\alpha_1+\alpha_2+\alpha_3=2d$. We also introduce the following notation for the above integral: 
	\begin{equation}\label{3fold-Selbergintegral}
		\GM_{\alpha_1,\alpha_2,\alpha_3}(x,y,z)\ =\
		\int_{\R^d}|x-t|^{-\alpha_2}|y-t|^{-\alpha_3}|z-t|^{-\alpha_1}\,dt,
	\end{equation}		
	and we shall suppress the parameters $\alpha_1$, $\alpha_2$ and $\alpha_3$ once they are fixed.
	\smallskip
	
	\noindent	
	Theorem~\ref{thm:main:2} extends~\eqref{eq:prelim:grafakos:formula} by incorporating angular factors. More precisely, identity~\eqref{eq:thm:MOZ:1} depends not only on the mutual distances among $x,y,z$, but also on their pairwise scalar products. The assumptions on the exponents are the same as in~\eqref{eq:prelim:grafakos:formula}. Thus, Theorem~\ref{thm:main:2} may be viewed as a natural generalisation of~\eqref{eq:prelim:grafakos:formula}. By contrast, explicit analogues for $k$-fold integrals with $k\geq 4$ are unavailable in general: see Wu \emph{et al.}~\cite{wu2020kfold}.
	
	\subsection{Fractional Laplacian: definition and properties}
	A standard definition of the fractional Sobolev space $H^s(\R^d)$, for $0<s<1$, is
	\begin{equation}\label{eq:prelim:def:Sobolev}
		H^s(\R^d)\ =\
		\left\{u\in L^2(\R^d):\;\intd(1+|\xi|^{2s})|\F u(\xi)|^2\,d\xi <\infty\right\}.
	\end{equation}
	Here $\F$ denotes the Fourier transform convention used in Frank-Lieb-Seiringer~\cite{frank2008hardy}, namely
	\begin{equation*}
		\F u(\xi)\ :=\
		\dfrac{1}{(2\pi)^{\frac d2}}\intd e^{-i\xi\cdot x}u(x)\,dx\,,
	\end{equation*}
	defined first for $u\in\calS(\R^d)$ and then extended to $H^s(\R^d)$. For $u \in H^s(\R^d)$, one has the identity
	\begin{equation}\label{eq:prelim:identity:1}
		\intd |\xi|^{2s}|\F u(\xi)|^2\,d\xi\ =\
		a_{s,d}\intd\intd \frac{|u(x)-u(y)|^2}{|x-y|^{d+2s}}\,dx\,dy
	\end{equation}
	where, by~Frank \emph{et al.}~\cite[Lemma 3.1]{frank2008hardy},
	\begin{equation}\label{eq:prelim:def:FourierConstant}
		a_{s,d}\ :=\
		\frac{s\,2^{2s-1}}{\pi^{\frac d2}}\frac{\Gamma(\frac{d+2s}{2})}{\Gamma(1-s)}.
	\end{equation} 
	Since throughout this paper we use the Fourier transform convention~\eqref{eq:prelim:def:Fourier}, we first relate these two normalisations. A direct computation shows that
	\begin{equation}\label{rltnFT}
		\F u \ =\ (2 \pi)^{\frac d2}\,\widehat{ \delta^{2 \pi} u }\quad\text{ for }u \in H^s(\R^d),
	\end{equation}
	where $\delta^a$ denotes the dilation operator introduced in~\eqref{eq:def:trans-dil-refl-function}. Substituting~\eqref{rltnFT} into~\eqref{eq:prelim:identity:1} and performing the corresponding change of variables, we obtain 
	\begin{equation}\label{eq:prelim:identity:2}
		(2\pi)^{2s} \intd |\xi|^{2s}|\widehat{v}(\xi)|^2d\xi 
		\ =\  a_{s,d} \intd\intd \frac{|v(x')-v(y')|^2}{|x'-  y'|^{d+2s}}\, dx'\,dy',
	\end{equation}
	for all $v\in H^s(\R^d)$. We define the fractional Laplacian $(-\Delta)^s$ as the pseudo-differential operator acting on $H^s(\R^d)$ with Fourier symbol $(2\pi)^{2s}|\xi|^{2s}$; see~\cite[Proposition 3.3]{nezza2012hitchhikers}. Thus,
	\begin{equation}\label{Fourierdef-fraclap}
		\wwidehat{\sLap u}(\xi)\ =\ (2\pi)^{2s}\,|\xi|^{2s}\,\hat{u}(\xi),
	\end{equation}
	or equivalently,             
	\begin{equation}\label{eq:prelim:def:quadratic}
		\inn{\sLap u}{\varphi}\ :=\ 
		(2\pi)^{2s}
		\intd|\xi|^{2s}\uhat(\xi)\overline{\phihat(\xi)}\,d\xi \quad \text{ for }\;\varphi\in H^{s}(\R^d).
	\end{equation}
	In particular, taking $\varphi = u$ and using~\eqref{eq:prelim:identity:2}, we obtain the quadratic form identity
	\begin{align*}
		\inn{\sLap u}{u}
		\ =\ (2\pi)^{2s}
		\intd|\xi|^{2s}|\uhat(\xi)|^2 \,d\xi
		\ =\  a_{s,d}\intd\intd\frac{|u(x)-u(y)|^{2}}{|x-y|^{d+2s}}\,dx\,dy\,.
	\end{align*}
	By polarisation, this yields the variational identity
	\begin{equation}\label{eq:prelim:varDef}
		\inn{\sLap u}{\varphi}\ =\
		a_{s,d}\intd\intd\frac{(u(x)-u(y))(\varphi(x)-\varphi(y))}{|x-y|^{d+2s}}\,dx\,dy
	\end{equation}
	for any $u,\varphi\in H^s(\R^d)$. This identity also provides the distributional definition of $\sLap u$ for more general $u$, by testing against $\varphi\in C^{\infty}_c(\R^d)$ whenever the right-hand side is finite. 
	\begin{remark}
		For $u\in\calS(\R^d)\cup C^2_c(\R^d)$, the above definition agrees with the principal value formula
		\begin{equation}\label{eq:prelim:def:pointwise:sLap}
			\sLap u(x)\ :=\
			2a_{s,d} \;{\rm P.V.}\intd\frac{u(x)-u(y)}{|x-y|^{d+2s}}\,dy\,,
		\end{equation}
		where $a_{s,d}$ is given by~\eqref{eq:prelim:def:FourierConstant}.
	\end{remark}
	In view of the Fourier transform law \eqref{eq:Fourier:Lieb-Loss:aux:1} and the Fourier definition of the fractional Laplacian \eqref{Fourierdef-fraclap}, for $0<\alpha<d-2s$, one has
	\begin{equation}\label{fracLap-Rszpot-alpha}
		\sLap|x|^{-\alpha} \ =\ \kappa_{d,s}(\alpha)\, |x|^{-\alpha-2s}\;\;\text{ in }\;\calS'(\R^d), \qquad \kappa_{d,s}(\alpha):= (2\pi)^{2s} \dfrac{\CLL{d-\alpha}\CLL{\alpha+2s}}{\CLL{\alpha}\CLL{d-\alpha-2s}}\,,
	\end{equation}
	in the sense of tempered distributions. In particular, taking $\alpha=\frac{d-2s}{2}$, it follows that
	\begin{equation}\label{fracLap-Rszpot}
		(-\Delta_1)^s\Bigl(|\,\cdot-x_2|^{-\frac{d-2s}{2}}\Bigr)\ =\
		\CFH \,|\,\cdot-x_2|^{-\frac{d+2s}{2}}\,\;\;\text{ in }\;\calS'(\R^d),\quad 0<s<\frac{d}{2}\,,
	\end{equation}
	since $\kappa_{d,s}\left((d-2s)/2\right) \ =\ \CFH$. This can also be obtained as a direct consequence of Frank \emph{et al.}~\cite[Equation~(3.4)]{frank2008hardy}. This will be used repeatedly in the proof of Theorem~\ref{thm:main:1}.
	\medskip	
	
	\noindent
	The following lemma is a basic ground-state estimate which will be useful for bounding the kinetic energy from below by a potential term. 
	\begin{lemma}[\textbf{Extraction of potential}]\label{lem:quadIneq}
		Let $0 < s < 1$ and  $u\in H^s(\R^d)$. Suppose that $\omega:\R^d\to\R_+$ satisfies $\omega\in H^s(\R^d)$ and $(u^2/\omega)\in H^s(\R^d)$. Then
		\begin{equation}\label{energy-ineq}
			\inn{\sLap u}{u}\ \geq\ \inn{\sLap\omega}{\frac{u^2}{\omega}}.
		\end{equation}
		Moreover, equality holds if and only if $u$ is a constant multiple of $\omega$.
	\end{lemma}
	\begin{proof}[Proof of Lemma~\ref{lem:quadIneq}]
		By the quadratic form identity,
		\begin{equation*}
			\inn{\sLap u}{u}\ =\ a_{s,d}\intd\intd\frac{|u(x)-u(y)|^2}{|x-y|^{d+2s}}\,dxdy\,.
		\end{equation*}
		We use the elementary identity
		\[
		|a-b|^2\ =\ (c-d)\left(\frac{a^2}{c}-\frac{b^2}{d}\right)+\left|\frac{a}{c}-\frac{b}{d}\right|^2cd\,,
		\]
		valid for $a,b,c,d\in\R$ with $c\neq 0$ and $d\neq 0$.
		Setting 
		\[
		a=u(x), \quad b=u(y),\quad c=\omega(x),\quad d=\omega(y),
		\]    
		we obtain
		\begingroup
		\allowdisplaybreaks
		\begin{align*}
			\inn{\sLap u}{u} 
			&\ =\ a_{s,d}\intd\intd\frac{(\omega(x)-\omega(y))(\frac{u^2}{\omega}(x)-\frac{u^2}{\omega}(y))}{|x-y|^{d+2s}}\,dx\,dy \\[0.25em]
			&\hspace{7em} + a_{s,d}\intd\intd\frac{\left|\frac{u}{\omega}(x)-\frac{u}{\omega}(y)\right|^2\omega(x)\omega(y)}{|x-y|^{d+2s}}\,dx\,dy\,.
		\end{align*}
		\endgroup
		The first integral is well defined by the assumptions $\omega\in H^s(\R^d)$ and $(u^2/\omega)\in H^s(\R^d)$. Hence, from the above identity we conclude that the second integral is also finite. Since $\omega>0$, the second integral is non-negative, and~\eqref{energy-ineq} follows. Equality holds if and only if this second integral vanishes, that is, if and only if $u/\omega$ is constant. 
	\end{proof}
	
	\begin{lemma}[\textbf{Localised version of Lemma~\ref{lem:quadIneq}}]\label{lem:localised:quadIneq}
		Let $0<s<1$. Let $U\subset\R^d$ be open, let $u\in C^{\infty}_c(U)$, and let $\omega\in H^s(\R^d)$ be non-negative a.e. in $\R^d$. Assume that $\omega\in C^{\infty}(U)$
		and that $\omega$ is strictly positive on a neighbourhood of $\supp u$.	Then,
		\[
		\inn{\sLap u}{u} \ \geq \ \inn{\sLap\omega}{\frac{u^2}{\omega}}.
		\]
	\end{lemma}
	
	\begin{proof}[Proof of Lemma~\ref{lem:localised:quadIneq}]
		Let $\rho(x):=e^{-|x|^2}$, and for $\epsilon>0$ set
		\[
		\omega_{\epsilon} \ :=\ \omega + \epsilon\,\rho.
		\]
		Then $\omega_{\epsilon}\in H^s(\R^d)$ and $\omega_{\epsilon}>0$ on $\R^d$. Since $u\in C^{\infty}_c(U)$, $\omega_{\epsilon}\in C^{\infty}(U)$ and $\omega_{\epsilon}$ is strictly positive on a neighbourhood of $\supp u$, it follows that 
		\[
		\frac{u^2}{\omega_{\epsilon}}\in C^{\infty}_c(U)\subset H^s(\R^d).
		\]
		Applying Lemma~\ref{lem:quadIneq} to $\omega_{\epsilon}$, we obtain
		\[
		\inn{\sLap u}{u} \ \geq \ \inn{\sLap\omega_{\epsilon}}{\frac{u^2}{\omega_{\epsilon}}}.
		\]
		Now
		\[
		\sLap\omega_{\epsilon} \ =\ \sLap\omega + \epsilon\sLap\rho \quad\text{ in }\,\calS'(\R^d),
		\]      
		and since $\omega$ is smooth and strictly positive on a neighbourhood of $\supp u$,
		\[
		\frac{u^2}{\omega_{\epsilon}}
		\longrightarrow \frac{u^2}{\omega}\quad\text{ in } C^{\infty}_c(U)\quad \text{ as } \epsilon\downarrow 0.
		\]
		Therefore, taking the limit $\epsilon\downarrow 0$ we conclude
		\[
		\inn{\sLap\omega_{\epsilon}}{\frac{u^2}{\omega_{\epsilon}}}\longrightarrow \inn{\sLap\omega}{\frac{u^2}{\omega}}.
		\]
		This proves the lemma.
	\end{proof}
	\noindent
	We shall also use the so-called Leibniz rule due to Biccari-Warma-Zuazua~\cite{biccari2017local}.
	\begin{proposition}\label{prop:LeibnitzRule}
		Let $u,v\in H^s(\R^d)$ and assume additionally that $uv\in H^s(\R^d)$. Then,
		\begin{align*}
			\sLap(uv)\ =\ v\sLap u+u\sLap v - \Nn(u,v)\,,
		\end{align*}
		in the variational sense~\eqref{eq:prelim:varDef}, where
		\begin{align}\label{interaction}
			\Nn(u,v)(x)\ :=\ 2a_{s,d}\intd\frac{(u(x)-u(y))(v(x)-v(y))}{|x-y|^{d+2s}}\,dy 
		\end{align}
		denotes the nonlocal interaction term of $u$ and $v$.
	\end{proposition}
	\begin{proof}[Proof of Proposition~\ref{prop:LeibnitzRule}]
		We prove the identity in the variational sense of~\eqref{eq:prelim:varDef}. Consider $\varphi\in \calS(\R^d)$. Since $u,v\in H^s(\R^d)$ one has $u\varphi,v\varphi\in H^s(\R^d)$, and the Cauchy-Schwarz inequality shows that 
		\begin{align*}
			|\inn{\Nn(u,v)}{\varphi}|\ \le\ 2a_{s,d}\,\|\varphi\|_{L^{\infty}(\R^d)}\,\intd \intd\frac{|u(x)-u(y)|\,|v(x)-v(y)|}{|x-y|^{d+2s}}\,dx\,dy < \infty\,.
		\end{align*}
		Therefore,
		\begingroup
		\allowdisplaybreaks 
		\begin{align}
			&\langle v\sLap u+u\sLap v - \Nn(u,v),\varphi\rangle \notag\\
			&\hspace{9em}\ =\ \langle \sLap u, v\varphi\rangle+\langle \sLap v, u\varphi\rangle-\langle \Nn(u,v),\varphi\rangle\,. \label{rhs}
		\end{align}
		\endgroup
		The last term can be rewritten in symmetric form as
		\begingroup
		\allowdisplaybreaks 
		\begin{align*}
			\langle \Nn(u,v),\varphi\rangle
			&\ =\ 2a_{s,d}\intd\biggl(\intd\frac{(u(x)-u(y))(v(x)-v(y))}{|x-y|^{d+2s}}dy\biggr)\varphi(x)\, dx\\
			&\ =\ a_{s,d}\intd \intd\frac{(u(x)-u(y))(v(x)-v(y)) (\varphi(x)+\varphi(y))} {|x-y|^{d+2s}}\,dy \, dx.
		\end{align*}
		\endgroup
		Therefore, using the definition \eqref{eq:prelim:varDef} for the first two terms in \eqref{rhs} and after some algebraic simplifications we obtain
		\[
		\begin{aligned}        
			&\langle \sLap u, v\varphi\rangle+\langle \sLap v, u\varphi\rangle-\langle \Nn(u,v),\varphi\rangle\\[0.25em]
			&\hspace{8em}\ =\ a_{s,d}\intd\intd\frac{(u(x)v(x)-u(y)v(y))(\varphi(x)-\varphi(y))}{|x-y|^{d+2s}}\,dy\,dx.
		\end{aligned}
		\]
		Since $uv\in H^s(\R^d)$, the right-hand side is precisely
		\[
		\inn{\sLap(uv)}{\varphi},
		\]
		which proves the claim.
	\end{proof}
	\subsection{A geometric identity}\label{subsec:geometric information}
	For any triangle with vertices at $x_i, x_j, x_k \in\R^d$, we shall denote by $r_{ij}$, $r_{jk}$, $r_{ki}$ the corresponding side lengths $|x_i-x_j|$, $|x_j-x_k|$, $|x_k-x_i|$, respectively. A key geometric identity relating the circumradius $R_{ijk}$ to the side lengths and the cosines of the corresponding angles (see~\cite[Lemma 3.2]{hoffmann2008many}) is
	\begin{equation}\label{circumradius}
		\frac{1}{2 R^2_{ijk}} \ =\ \sum_{\mbox{\rm cyclic}} \frac{(x_k-x_i)}{|x_k-x_i|^2}\cdot\frac{(x_j-x_i)}{|x_j-x_i|^2} \,.
	\end{equation}
	\medskip
	
	\section{Proof of Theorem~\ref{thm:main:1} and further remarks}\label{sec:proof:main_thm}
	The results of this section rely on the ground-state representation technique; see, for instance, Lundholm~\cite[Section 4]{lundholm2017methods}. We begin by briefly recalling the corresponding argument in the local setting before turning to the fractional framework. 
	Let $\Omega$ be a domain in $\R^d$ and let $\omega:\Omega\to\R_+$ be a positive twice-differentiable weight. Given $u\in C^{\infty}_c(\Omega)$, write $u=:\omega\, v$. A direct computation gives
	\[
	|\nabla u|^2\ =\ v^2|\nabla\omega|^2+\omega\nabla\omega\cdot\nabla(v^2)+\omega^2|\nabla v|^2.
	\]
	Integrating over $\Omega$ and applying integration by parts, one obtains
	\begin{equation}
		\int_{\Omega}|\nabla u|^2\,dx
		\ =\
		\int_{\Omega}\left(|\omega\nabla v|^2+\frac{(-\Delta\omega)}{\omega}
		u^2\right)dx\label{eq:intro:aux:1}
	\end{equation}
	In particular, one obtains
	\begin{equation}        
		\inn{-\Delta u}{u}_{L^2(\R^d)} \ =\ \int_{\Omega}|\nabla u|^2\,dx
		\ \geq\ \int_{\Omega}\frac{(-\Delta\omega)}{\omega}\,u^2\,dx.\label{eq:intro:aux:2}
	\end{equation}
	This identity is referred to as the ground-state representation (GSR). It is useful because it provides a lower bound for the kinetic energy in terms of a potential term involving 
	\[
	V(x)\ :=\ \frac{-\Delta\omega(x)}{\omega(x)}.
	\]    
	In the case $\Omega=\R^d$, choosing $\omega(x)=|x|^{-\alpha}$ and optimising with respect to $\alpha$ yields the classical Hardy inequality~\eqref{eq:intro:hardyIneq}. 
	
	\medskip
	
	\subsection{Proof of Theorem~\ref{thm:main:1}}
	
	We now implement the ground-state representation method for three particles in the fractional setting. In this setting, however, the same strategy becomes substantially more delicate. The proof of Theorem~\ref{thm:main:1} relies on three ingredients: a truncated three-particle ground state, the fractional Leibniz rule, and an explicit evaluation of the resulting nonlocal interaction term. The choice of the ground state
	\[
	|x_i-x_j|^{-\frac{(d-2s)}{2}}|x_i-x_k|^{-\frac{(d-2s)}{2}}|x_j-x_k|^{-\frac{(d-2s)}{2}}
	\]
	is straightforward and is inspired by Lundholm~\cite{lundholm2015geometric} in the local setting.
	However, since it is not directly admissible in Lemma~\ref{lem:quadIneq}, we introduce a truncation argument using suitable cut-off functions. With this choice made, we estimate the kinetic energy associated with each particle by means of a fractional Leibniz rule. Finally, these steps lead to the analysis of a nonlocal interaction term, whose explicit evaluation is carried out in Lemma~\ref{lem:limit:interaction} and relies on the new angular Selberg-type identity of Theorem~\ref{thm:main:2}.\medskip
	
	\noindent\textbf{Step 1: Choice of the truncated ground state.}\\ 
	Let $\phi\in C^{\infty}_c(\R^d)$ be a radial function such that $0\leq\phi\leq1$ and $\phi(x)\equiv 1$ for $|x|\leq 1$. For $0<\alpha<(d-2s)/2$ and $n\in\N$,
	define $\valphan:\R^d\to\R$ via
	\begin{equation}\label{eq:proof:main:1:GSR:1}
		\valphan(x)\ :=\ |x|^{-\alpha}\phi\Bigl(\frac{x}{n}\Bigr) \quad\text{ for }\,x\in\R^d. 
	\end{equation}
	It is known that $\valphan\in H^s(\R^d)$; see Frank-Lieb-Seiringer~\cite[Proposition 4.1]{frank2008hardy}. Because the kinetic energy is estimated particle by particle, for indices $i,j,k \in \{1,2,3\}$ with $i\neq j$ and $i\neq k$, and for $(x_i,x_j,x_k)\in\off{3}$, we define the truncated three-particle ground state associated with the $i$-th particle as follows:
	\begin{equation}\label{eq:proof:main:1:GSR:2}
		\om{i}(x_i)\ :=\ \valphan(x_i-x_j)\,\valphan(x_i-x_k)\,|x_j-x_k|^{-\alpha}.
	\end{equation}
	By Lemma~\ref{lem:app:regularity} in the appendix, one has $\om{i}\in H^s(\R^d)$ for any admissible choice of $i$, $\alpha$, and $n$.\medskip
	
	\noindent\textbf{Step 2: Application of the localised ground-state estimate.}\\
	Let $u\in C^{\infty}_c(\off{3})$. For each $i=1,2,3$, we regard $u$ as a function of the variable~$x_i$, with the remaining two variables held fixed, and denote this partial function by $u_i:\R^d\to\R$. 
	
	Fix $i=1$ and $x_2\neq x_3$. Since $u\in C^{\infty}_c(\off{3})$, there exist $R,\delta>0$, depending on $(x_2,x_3)$, such that  
	\[
	\supp u_1 \ \subset\ B_R(0)\,\setminus\,\bigl(B_{\delta}(x_2)\cup B_{\delta}(x_3)\bigr).
	\]
	Choose $n$ sufficiently large such that
	\[
	|x_1-x_2|<n,\qquad |x_1-x_3|<n
	\]
	for every $x_1\in\supp u_1$. Since $\phi\equiv 1$ on $B_1(0)$, it follows for large $n$ that
	\[
	\valphan(x_1-x_2)\ =\ |x_1-x_2|^{-\alpha},\qquad \valphan(x_1-x_3) \ =\ |x_1-x_3|^{-\alpha},
	\]
	throughout $\supp u_1$. Since $\om{1} \ =\ \valphan(x_1-x_2)\,\valphan(x_1-x_3)\,|x_2-x_3|^{-\alpha}$ is smooth and strictly positive on a neighbourhood $U$ of $\supp u_1$, it follows that
	\[
	\frac{u^2_1}{\om{1}}\in C^{\infty}_c(U)\subset H^s(\R^d).
	\]
	Therefore, Lemma~\ref{lem:localised:quadIneq} applies and yields
	\[
	\inn{\sLap u_1}{u_1} \ \geq\ \inn{\sLap\om{1}}{\frac{u^2_1}{\om{1}}}. 
	\]
	The cases $i=2$ and $i=3$ follow by relabelling.\medskip
	
	\noindent\textbf{Step 3: Fractional Leibniz decomposition.}\\
	Since the factor $|x_2-x_3|$ is independent of $x_1$, it plays no role when the operator $\sLapj{1}$ acts, and we may therefore focus on the product
	\[
	\valphan(\,\cdot-x_2)\,\valphan(\,\cdot-x_3).  
	\]
	Throughout this step, the quotients appearing below make sense on the neighbourhood $U\supset\supp u_1$ introduced in Step 2. On $U$, $\valphan(\,\cdot-x_2)$ and $\valphan(\,\cdot-x_3)$ are smooth and strictly positive.
	
	Applying the fractional Leibniz rule (Proposition~\ref{prop:LeibnitzRule}) to the above product, we obtain a decomposition into two two-body terms and one interaction term:
	\begingroup
	\allowdisplaybreaks
	\begin{align*}
		&\inn{\sLapj{1}\om{1}}{\frac{u^2_1}{\om{1}}}\\
		&\hspace{3em} =\,
		\inn{\sLapj{1}\valphan(\,\cdot-x_2)}{\frac{|u(\cdot,x_2,x_3)|^2}{\valphan(\,\cdot-x_2)}}\ +\ \inn{\sLapj{1}\valphan(\,\cdot-x_3)}{\frac{|u(\cdot,x_2,x_3)|^2}{\valphan(\,\cdot-x_3)}}\\[0.15em]
		&\hspace{4em} -\,\int_{\R^{d}}\frac{\Nn(\valphan(\,\cdot-x_2),\valphan(\,\cdot-x_3))(x_1)}{\valphan(x_1-x_2)\valphan(x_1-x_3)}\,|u(x_1,x_2,x_3)|^2\,dx_1.
	\end{align*}
	\endgroup
	
	\noindent\textbf{Step 4: Passage to the limit}.\\
	By Lemma~\ref{lem:app:LapIPP:2}, after passing first to the limit $n \to \infty$ and then to $\alpha\uparrow(d-2s)/2$, the first two terms converge to the expected pairwise Hardy contributions:
	\[
	\CFH \intd \bigl(|x_1-x_2|^{-2s}+|x_1-x_3|^{-2s}\bigr)\,|u(x_1,x_2,x_3)|^2\,dx_1\,.
	\]
	The interaction term is the only point at which the restriction $d\geq 4-2s$ enters, through the explicit evaluation provided by Lemma~\ref{lem:limit:interaction}. Thus, it remains to identify the limit of the interaction term. The analytic bottleneck is resolved through an essential use of the angular correlation identity~\eqref{eq:thm:MOZ:1} in Theorem~\ref{thm:main:2}. Then, by Lemma~\ref{lem:limit:interaction}, the limit
	\[
	\lim_{\alpha\uparrow\frac{d-2s}{2}}\lim_{n\to\infty}\frac{\Nn\bigl(\valphan(\,\cdot-x_2),\valphan(\,\cdot-x_3)\bigr)(x_1)}{\valphan(x_1-x_2)\,\valphan(x_1-x_3)}
	\]
	is given by 
	\[
	\begin{aligned}
		& \CFH\biggl[\frac{1}{|x_1-x_2|^{2s}}+\frac{1}{|x_1-x_3|^{2s}}+2\frac{(x_3-x_1)}{|x_3-x_1|^{2s}}\cdot \frac{(x_2-x_1)}{|x_2-x_1|^{2s}}\frac{1}{|x_2-x_3|^{2-2s}}\\
		& -\frac{|x_1-x_3|^{2-2s}}{|x_1-x_2|^{2s}|x_2-x_3|^{2-2s}}-\frac{|x_1-x_2|^{2-2s}}{|x_1-x_3|^{2s}|x_2-x_3|^{2-2s}}\biggr]
	\end{aligned}
	\]
	Substituting this into the preceding expression and simplifying, we obtain
	\begingroup
	\allowdisplaybreaks
	\begin{align*}\label{eq:proof:main:1:aux:2}
		&\inn{\sLapj{1}u_1}{u_1}\notag\\
		&\hspace{4em} \geq\ \CFH \int_{\R^{d}} \biggl[\frac{|x_1-x_2|^{2-2s}}{|x_1-x_3|^{2s}|x_2-x_3|^{2-2s}}+\frac{|x_1-x_3|^{2-2s}}{|x_1-x_2|^{2s}|x_2-x_3|^{2-2s}}\\[0.25em]
		&\hspace{11em}- 2\frac{(x_3-x_1)}{|x_3-x_1|^{2s}}\cdot\frac{(x_2-x_1)}{|x_2-x_1|^{2s}}\frac{1}{|x_2-x_3|^{2-2s}}\biggr] |u(x_1,x_2,x_3)|^2\,dx_1.
	\end{align*}
	\endgroup
	\smallskip
	
	\noindent\textbf{Step 5: Identification of the three-body potential.}\\
	We now invoke the law of cosines, $|x_1-x_2|^2+|x_1-x_3|^2-2(x_2-x_1)\cdot (x_3-x_1) = |x_2-x_3|^2$ to obtain a key simplification. Indeed, the expression in brackets simplifies to
	\[
	\frac{|x_2-x_3|^{2s}}{|x_1-x_2|^{2s}|x_1-x_3|^{2s}}.
	\]
	Consequently,
	\[
	\inn{\sLapj{1}u_1}{u_1} \ \geq\  \CFH \int_{\R^{d}} \frac{|x_2-x_3|^{2s}}{|x_1-x_2|^{2s}|x_1-x_3|^{2s}}\,|u(x_1,x_2,x_3)|^2\,dx_1.
	\]
	The corresponding estimates for $i=2$ and $i=3$ follow analogously, yielding, respectively 
	\[
	\inn{\sLapj{2}u_2}{u_2} \ \geq\  \CFH \int_{\R^{d}} \frac{|x_1-x_3|^{2s}}{|x_2-x_1|^{2s}|x_2-x_3|^{2s}}\,|u(x_1,x_2,x_3)|^2\,dx_2,
	\]
	and
	\[
	\inn{\sLapj{3}u_3}{u_3} \ \geq\  \CFH \int_{\R^{d}} \frac{|x_1-x_2|^{2s}}{|x_3-x_1|^{2s}|x_3-x_2|^{2s}}\,|u(x_1,x_2,x_3)|^2\,dx_3.
	\]
	Integrating each of these inequalities with respect to the remaining variables and summing, one obtains the lower bound in Theorem~\ref{thm:main:1},    \begin{equation}\label{eq:proof:main:reducedPot}
		\V(x_1,x_2,x_3)\ =\ \sum_{\mbox{\rm cyclic}}\frac{|x_k-x_j|^{2s}}{|x_k-x_i|^{2s}|x_j-x_i|^{2s}}\,.
	\end{equation}
	This is precisely inequality~\eqref{eq:thm:1:inequality}, which completes the proof of Theorem~\ref{thm:main:1}. \qed
	\medskip
	
	\subsection{Sharpness of the coefficient for the fixed potential $\V$}
	
	In this subsection we prove that the coefficient  \emph{$\CFH$ in~\eqref{eq:thm:1:inequality} is sharp for the explicit potential $\V$}. In other words, sharpness is understood relative to the fixed potential $\V$, and not with respect to the fully optimised three-particle Hardy problem. To this end, define 
	\begin{equation}\label{eq:sharp:Vs3:def}
		\optC
		\ :=\
		\inf_{u\in C_c^\infty(\off{3})\setminus\{0\}}
		\frac{\displaystyle \sum_{i=1}^3 \left\langle \sLapj{i}u,u\right\rangle_{L^2(\R^{3d})}}
		{\displaystyle \int_{\R^{3d}}\V(x)\,|u(x)|^2\,dx }.
	\end{equation}
	By Theorem~\ref{thm:main:1}, one has $\optC\geq\CFH$. It therefore remains to prove the reverse inequality.
	
	The sharpness mechanism is as follows. We construct test functions for which two particles approach each other at a small scale $\varepsilon$, while the third particle remains at a fixed macroscopic distance. In this regime, the three-body potential $\V$ reduces, to leading order, to twice the two-body Hardy potential, and the corresponding kinetic energy exhibits the same asymptotic behaviour. Consequently, the three-particle quotient reduces to the two-particle Hardy quotient.  
	
	Let \(\{w_n\}^{\infty}_{n=1}\subset C_c^\infty(\R^d\setminus\{0\})\) be a minimising sequence for the fractional Hardy constant \(\CFH\), that is
	\[
	E_n:=\innd{\sLap w_n}{w_n},
	\quad H_n:=\int_{\R^d}\frac{|w_n(r)|^2}{|r|^{2s}}\,dr,
	\quad\frac{E_n}{H_n}\to \CFH.
	\]
	Choose $\varphi,\psi\in C_c^\infty(\R^d)$ such that
	\[
	\supp\varphi,\;\supp\psi\subset B_1(0),  \qquad\|\varphi\|_{L^2(\R^d)}=\|\psi\|_{L^2(\R^d)}=1,
	\]
	and fix \(a\in\R^d\) with \(|a|=10\). For \(n\in\N\) and \(\varepsilon>0\), set
	\[
	u_{n,\varepsilon}(x_1,x_2,x_3)
	\ :=\ \varepsilon^{-\frac{d}{2}}\,w_n\!\left(\frac{x_1-x_2}{\varepsilon}\right)\varphi\!\left(\frac{x_1+x_2}{2}\right)
	\psi(x_3-a).
	\]
	Fix $n$ and consider the asymptotic behaviour as $\epsilon\to0$. Since \(\supp w_n\subset\subset\R^d\setminus\{0\}\), there exist \(0<r_n<R_n<\infty\) such that
	\[
	\supp w_n \ \subset\ \{\theta\in\R^d:\;r_n\leq |\theta|\leq R_n\}.
	\]
	Introduce the change of variables
	\[
	y\ :=\ \frac{x_1+x_2}{2},\quad \theta\ :=\ \frac{x_1-x_2}{\varepsilon},\quad z\ :=\ x_3-a.
	\]
	For fixed $(y,\theta,z)$ with $\theta\neq 0$, the first two particles coalesce as $\epsilon\downarrow 0$, whereas the third particle remains at a bounded distance and does not coalesce with the pair $(x_1,x_2)$. It follows that, for $0<\varepsilon<\varepsilon_n:=\min\{1,8/R_n\}$, one has $u_{n,\varepsilon}\in C^{\infty}_c(\off{3})$. Indeed, $x_1-x_2$ is of order $\epsilon$, the centre of mass $(x_1+x_2)/2$ remains in a compact set, and $x_3$ is localised near $a$, hence remains separated from the first two particles.
	
	Now define
	\[
	D_{n,\varepsilon}\ :=\ \int_{\R^{3d}}\V(x)|u_{n,\varepsilon}(x)|^2\,dx,\qquad Q_{n,\varepsilon}\ :=\ \sum^3_{i=1}\inn{\sLapj{i}u_{n,\varepsilon}}{u_{n,\varepsilon}}_{L^2(\R^{3d})}.
	\]
	We first analyse the denominator. For the above change of variables $(x_1,x_2,x_3)\mapsto (y,\theta,z)$, the volume element satisfies $dx_1\,dx_2\,dx_3 = \varepsilon^d\,dy\,d\theta\,dz$. Then,
	\[
	\varepsilon^{2s}D_{n,\varepsilon} \ =\ \int_{\R^{3d}}\varepsilon^{2s}\,\V\Bigl(y+\frac{\varepsilon\theta}{2},\,y-\frac{\varepsilon\theta}{2},\,a+z\Bigr)|w_n(\theta)|^2|\varphi(y)|^2|\psi(z)|^2\,dy\,d\theta\,dz.
	\]
	By direct use of the definition of $\V$ and a straightforward calculation,
	\[
	\varepsilon^{2s}\V\Bigl(y+\frac{\varepsilon\theta}{2},\,y-\frac{\varepsilon\theta}{2},\,a+z\Bigr)\longrightarrow\frac{2}{|\theta|^{2s}}\qquad\text{ as }\varepsilon\downarrow 0. 
	\]
	Moreover, on the support of the integrand, one has $|\theta|\geq r_n$ and $|a+z-y\pm\varepsilon\theta/2|\geq 4$, so that the integrand is dominated by
	\[
	C_n\,|w_n(\theta)|^2\,|\varphi(y)|^2\,|\psi(z)|^2,
	\]
	which is integrable. Hence, by the dominated convergence theorem,
	\[
	\lim_{\varepsilon\downarrow 0}\varepsilon^{2s}D_{n,\varepsilon} \ =\ 2\int_{\R^{3d}}\frac{|w_{n}(\theta)|^2}{|\theta|^{2s}}|\varphi(y)|^2\,|\psi(z)|^2\,dy\,d\theta\,dz \ =\ 2 H_n.
	\]
	We next turn to the numerator
	\[
	Q_{n,\epsilon}\ =\ (2\pi)^{2s}\int_{\R^{3d}}\bigl(|\xi|^{2s}+|\eta|^{2s}+|\zeta|^{2s}\bigr)\bigl|\wwidehat{u_{n,\epsilon}}(\eta,\xi,\zeta)\bigr|^2\,d\eta\,d\xi\,d\zeta.
	\]
	It can be checked that 
	\[
	\wwidehat{u_{n,\epsilon}}(\eta,\xi,\zeta) = \ \epsilon^{\frac{d}{2}}\,\wwidehat{w_n}\left(\epsilon \frac{(\eta - \xi)}{2}\right)\,\wwidehat{\varphi}(\xi+\eta)\,e^{-2\pi ia\cdot\zeta}\,\wwidehat{\psi}(\zeta).
	\]
	Introduce the change of variables
	\[
	v\ :=\ \xi+\eta,\quad u\ :=\ \varepsilon\left(\frac{\xi-\eta}{2}\right),\quad \zeta\ :=\ \zeta,
	\]
	so that the volume element satisfies $d\xi\,d\eta\,d\zeta = \varepsilon^{-d}\,du\,dv\,d\zeta$. Therefore, 
	\begingroup
	\allowdisplaybreaks
	\begin{align*}
		\epsilon^{2s}Q_{n,\epsilon} \ &=\ (2\pi)^{2s}\int_{\R^{3d}}\left(\left|u+\frac{\epsilon v}{2}\right|^{2s}+\left|u-\frac{\epsilon v}{2}\right|^{2s}+\epsilon^{2s}|\zeta|^{2s}\right)\\
		&\hspace{15em} \bigl|\wwidehat{w_n}(u)\bigr|^2\bigl|\wwidehat{\varphi}(v)\bigr|^2\bigl|\wwidehat{\psi}(\zeta)\bigr|^2\,du\, dv\, d\zeta.    
	\end{align*} 
	\endgroup 
	The integrand converges pointwise to 
	\[
	2(2\pi)^{2s}|u|^{2s}\bigl|\wwidehat{w_n}(u)\bigr|^2\bigl|\wwidehat{\varphi}(v)\bigr|^2\bigl|\wwidehat{\psi}(\zeta)\bigr|^2,
	\]
	and is dominated by an $L^1$-function, since
	\[
	\Bigl|u\pm\frac{\epsilon v}{2}\Bigr|^{2s}\leq C_s\bigl(|u|^{2s}+|v|^{2s}\bigr)\quad\text{ and }\quad \epsilon^{2s}|\zeta|^{2s}\leq |\zeta|^{2s}
	\]
	for $0<\epsilon<1$, while $\wwidehat{w_n}, \wwidehat{\varphi},\wwidehat{\psi}\in \calS(\R^d)$. Then the dominated convergence theorem yields
	\[
	\lim_{\epsilon\downarrow 0}\epsilon^{2s}Q_{n,\epsilon}\ =\ 2(2\pi)^{2s}\intd|u|^{2s}|\wwidehat{w_n}(u)|^2\,du\ =\ 2E_n.
	\]
	Thus, for each fixed $n\in\N$,
	\[
	\frac{Q_{n,\epsilon}}{D_{n,\epsilon}} \ =\ \frac{\epsilon^{2s}Q_{n,\epsilon}}{\epsilon^{2s}D_{n,\epsilon}} \longrightarrow \frac{2 E_n}{2 H_n}=\frac{E_n}{H_n}\qquad \text{ as }\epsilon\downarrow 0.
	\]
	Now, since $u_{n,\epsilon}\in C^{\infty}_c(\off{3})$ for $0<\epsilon<\epsilon_n$, the definition of $\optC$ yields
	\[
	\optC\ \leq\ \frac{Q_{n,\epsilon}}{D_{n,\epsilon}}.
	\]
	Letting $\epsilon\downarrow 0$ and then $n\to\infty$, we conclude that
	\[
	\optC\ \leq\ \frac{E_{n}}{H_{n}}\longrightarrow \CFH.
	\]
	Combining this with the lower bound from Theorem~\ref{thm:main:1}, we obtain
	\[
	\optC \ =\ \CFH.
	\]
	This proves the sharpness of the coefficient in Theorem~\ref{thm:main:1}.
	\medskip

	\subsection{Sharper forms of Theorem~\ref{thm:main:1} and geometric consequences}\label{sec:further-remarks}
	
	We now derive consequences of Theorem~\ref{thm:main:1} that clarify the geometric structure of the three-particle potential. The first provides a lower bound in terms of the circumradius of the triangle determined by the three particles. The second explicitly identifies the non-negative defect by which the three-body potential exceeds the standard pairwise Coulomb-type interaction.
	
	Let $a,b,c$ denote the side lengths of a triangle and let $R$ be its circumradius. Then one has
	\begin{equation}\label{eq:geom:basicIneq}
		\frac{1}{a^{2s}}+\frac{1}{b^{2s}}+\frac{1}{c^{2s}}
		\ \geq\ \frac{3^{1-s}}{R^{2s}}.
	\end{equation}
	Equality holds in both~\eqref{eq:lem:h-o:1} and~\eqref{eq:geom:basicIneq} if and only if the triangle is equilateral, that is, $a=b=c$.
	The proof is analogous to the case $s=1$ treated in Hoffmann-Ostenhof \emph{et al.}~\cite[Lemma 3.3]{hoffmann2008many}, and is therefore omitted. 
	
	As a direct consequence of Theorem~\ref{thm:main:1} combined with~\eqref{eq:lem:h-o:1} and \eqref{eq:geom:basicIneq}, we obtain the following geometric Hardy inequality. 	
	
	\begin{corollary}[\textbf{A geometric three-particle fractional Hardy inequality}]\label{cor:proof:Hoffmann-ostenhof}
		\begin{align*}
			&\intd\intd\intd u(x_1,x_2,x_3)\left(\sum^3_{i=1}\sLapj{i}u(x_1,x_2,x_3)\right)dx_1\,dx_2\,dx_3\\
			& \hspace{6em}\ \geq\ 3^{1-s} \CFH \intd\intd\intd \frac{|u(x_1,x_2,x_3)|^2}{(R_{123})^{2s}}\,dx_1\,dx_2\,dx_3,
		\end{align*}
		where $R_{123}$ is the circumradius of the triangle of vertices $x_1,x_2$ and $x_3$.  
	\end{corollary}
	
	\noindent
	A second consequence is that the improvement over the pairwise potential admits a fully explicit representation. Indeed, the algebraic identity 
	\begingroup
	\allowdisplaybreaks
	\begin{align*}
		&\frac{c^{2s}}{a^{2s}b^{2s}}+\frac{b^{2s}}{a^{2s}c^{2s}}+\frac{a^{2s}}{b^{2s}c^{2s}}\\
		&\hspace{8em}\ =\ \frac{1}{a^{2s}}+\frac{1}{b^{2s}}+\frac{1}{c^{2s}}+\frac{\bigl(a^{2s}-b^{2s}\bigr)^2+\bigl(b^{2s}-c^{2s}\bigr)^2+\bigl(c^{2s}-a^{2s}\bigr)^2}{2a^{2s}b^{2s}c^{2s}}    
	\end{align*}
	\endgroup
	shows, in particular, that~\eqref{eq:lem:h-o:1} holds and provides an explicit expression for its defect term. Consequently, for every $s\in(0,1)$ and every $(x_1,x_2,x_3)\in\off{3}$, the three-particle potential admits the pointwise decomposition
	\begin{equation}\label{eq:pot:remainder}
		\V(x_1,x_2,x_3)
		\ =\ \sum_{1\leq i<j\leq 3}\frac{1}{|x_i-x_j|^{2s}}+\rem(x_1,x_2,x_3),
	\end{equation}
	where
	\begin{equation}\label{eq:def:remainder}	    
		\displaystyle{\rem(x_1,x_2,x_3)
			\ :=\ \frac{\bigl(r^{2s}_{12}-r^{2s}_{13}\bigr)^2+\bigl(r^{2s}_{13}-r^{2s}_{23}\bigr)^2+\bigl(r^{2s}_{23}-r^{2s}_{12}\bigr)^2}
			{2\,r_{12}^{2s}\,r_{13}^{2s}\,r_{23}^{2s}}
			\ \ge\ 0.}
	\end{equation}
	Note that the remainder $\rem(x_1,x_2,x_3)$ vanishes if and only if the triple is equilateral. By summing the decomposition~\eqref{eq:pot:remainder} over all triples one obtains
	\begin{align}\label{eq:12}
		\int_{\R^{Nd}} u(x)&\left(\sum_{i=1}^N \sLapj{i}u(x)\right)dx\notag\\
		&\hspace{2em} \geq\ \frac{2\,\CFH}{N-1}\int_{\R^{Nd}}\Biggl(\sum_{\substack{1\le i<j\le N}}\frac{1}{|x_i-x_j|^{2s}}\Biggr)\,|u(x)|^2\,dx\notag\\
		&\hspace{3em}+\frac{2\,\CFH}{(N-1)(N-2)}\int_{\R^{Nd}}\Biggl(\sum_{\substack{1\le i<j<k\le N}}\rem(x_i,x_j,x_k)\Biggr)\,|u(x)|^2\,dx,
	\end{align}
	for every $u\in C^{\infty}_c(\off{N})$, which is a refinement of the elementary bound~\eqref{eq:intro:fracHardyQuadraticVersion}.
	
	\begin{proposition}\label{prop:uniform:domination}
		Let $d\geq 1$, $s\in (0,1)$, and let $N>d+1$. For any $x=(x_1,\ldots, x_N)\in \off{N}$, define
		\[
		\remR(x) \ :=\ \sum_{1\leq i<j<k\leq N} \rem(x_i,x_j,x_k). 
		\]
		Then there exists a constant $\delta=\delta(d,s,N)>0$ such that
		\[
		\remR(x) \ \geq\ \delta\sum_{1\leq i<j\leq N}\frac{1}{|x_i-x_j|^{2s}}\quad\text{ for every }x\in\off{N}.
		\]
	\end{proposition}
	\begin{proof}[Proof of Proposition~\ref{prop:uniform:domination}]
		Consider the infimum
		\[
		\inf_{\x\in\off{N}}\frac{\remR(x)}{\sum_{1\leq i<j\leq N}\bigl|x_i-x_j\bigr|^{-2s}}.
		\]
		We claim that the above infimum is strictly positive. Suppose, on the contrary, that it is not the case, and let $x^{(n)}=(x^{(n)}_1,\ldots, x^{(n)}_N)\in\off{N}$ be a minimising sequence. Then,  
		\[
		\frac{\remR(x^{(n)})}{\sum_{1\leq i<j\leq N}\bigl|x^{(n)}_i-x^{(n)}_j\bigr|^{-2s}}\longrightarrow0\qquad\text{ as } n\to\infty.
		\]
		Since the above ratio is invariant under translations, dilations and permutations of the $N$-tuple, we may assume without loss of generality that
		\[
		x^{(n)}_1\ =\ 0,\qquad |x^{(n)}_2|\ =\ 1,\qquad |x^{(n)}_i-x^{(n)}_j|\ \geq\ 1 \quad\text{ for all }1\leq i<j\leq N.
		\]
		Under this normalisation, the denominator is bounded above by $\binom{N}{2}$, and hence
		\begin{equation}\label{eq:rem:limit}
			\remR\bigl(x^{(n)}\bigr)\longrightarrow 0\qquad\text{ as }n\to\infty.
		\end{equation}
		Clearly, both $x^{(n)}_1$ and $x^{(n)}_2$ remain bounded. We claim that $\{x^{(n)}_j\}_n$ is bounded for all $j$. Suppose instead that $|x^{(n)}_m|\to\infty$ as $n\to\infty$ for some $m\notin\{1,2\}$. Set
		\[
		B_n\ :=\ |x^{(n)}_m|^{2s}\to\infty,\quad C_n\ :=\ |x^{(n)}_2-x^{(n)}_m|^{2s}\to\infty,\quad A_n \ :=\ |x^{(n)}_2|^{2s}=1 \quad\text{ for all }n\in\N.
		\]
		We focus on the term $\rem(x^{(n)}_1,x^{(n)}_2,x^{(n)}_m)$. By the definition of this remainder term,
		\begingroup
		\begin{align*}
			\rem(x^{(n)}_1,x^{(n)}_2,x^{(n)}_m)
			&\ =\ \frac{(1-B_n)^2+(B_n-C_n)^2+(1-C_n)^2}{2B_nC_n}\\[0.25em]
			&\ =\ \frac{B_n}{C_n}+\frac{C_n}{B_n}-1-\frac{1}{B_n}-\frac{1}{C_n}+\frac{1}{B_nC_n}.
		\end{align*}
		\endgroup
		Since $t+1/t\geq 2$ for every $t>0$, it follows that
		\[
		\rem(x^{(n)}_1,x^{(n)}_2,x^{(n)}_m)
		\ \geq\ 1-\frac{1}{B_n}-\frac{1}{C_n}+\frac{1}{B_nC_n}.
		\]
		Therefore, $\rem(x^{(n)}_1,x^{(n)}_2,x^{(n)}_m)$ remains bounded away from zero as $n\to\infty$, contradicting~\eqref{eq:rem:limit}.
		
		It follows that all sequences $\{x^{(n)}_m\}^{\infty}_{n=1}$ are bounded. Passing to a subsequence, we may therefore assume that
		\[
		x^{(n)}\longrightarrow x^{\star}=(x^{\star}_1,\ldots,x^{\star}_N)\in(\R^d)^N \qquad\text{ as }n\to\infty.
		\]
		In particular, $x^{\star}\in\off{N}$ since all pairwise distances are bounded below by $1$. Passing to the limit as $n\to\infty$, we deduce $\remR(x^{\star})=0$. Consequently, each summand $\rem(x^{\star}_i,x^{\star}_j,x^{\star}_k)=0$, and hence every triple $x^{\star}_i$,$x^{\star}_j$, $x^{\star}_k$ forms an equilateral triangle. It follows that $x^{\star}_1,x^{\star}_2,\ldots, x^{\star}_N$ are equidistant from each other in $\R^d$. This is only possible if $N\leq d+1$, which contradicts the assumption $N>d+1$. This completes the proof.
	\end{proof}
	
	\begin{remark}[\textbf{On the condition $\mathbf{N>d+1}$}]\label{rem:necesity:condition}
		The restriction $N>d+1$ is sharp for Proposition~\ref{prop:uniform:domination}. Indeed, if $N\leq d+1$, one may choose $x_1,\ldots, x_N$ as the vertices of a regular simplex in $\R^d$. Then, every triple is equilateral, hence $\rem(x_i,x_j,x_k)=0$ for all $1\leq i<j<k\leq N$, and therefore $\remR(x)=0$, while $\sum_{i<j}|x_i-x_j|^{-2s}>0$. Thus, no positive constant $\delta$ can satisfy the conclusion of Proposition ~\ref{prop:uniform:domination} in that range.
	\end{remark}
	
	\begin{corollary}[\textbf{Improved Coulomb-type constant}]\label{cor:strict:improvement}
		Let $d\geq 1$, $s\in (0,1)$, and $N\geq 3$, and assume $N>d+1$. Then there exists $\delta=\delta(d,s,N)>0$ such that
		\begingroup
		\allowdisplaybreaks
		\begin{align*}    
			& \int_{\R^{Nd}}u(x)\left(\sum^N_{i=1}\sLapj{i}u(x)\right)\,dx\\
			& \hspace{5em}\geq \frac{2\,\CFH}{N-1} \left(1+\frac{\delta}{N-2}\right)\int_{\R^{Nd}}\left(\sum_{1\leq i<j\leq N}\frac{1}{|x_i-x_j|^{2s}}\right)|u(x)|^2\,dx\,,
		\end{align*}
		\endgroup
		for every $u\in C^{\infty}_c(\off{N})$.
	\end{corollary}
	\begin{proof}[Proof of Corollary~\ref{cor:strict:improvement}]
		By the refined many-particle inequality derived above,
		\begingroup
		\allowdisplaybreaks
		\begin{align*}
			\int_{\R^{Nd}}u(x)\left(\sum^N_{i=1}\sLapj{i}u(x)\right)\,dx
			&\geq \frac{2\,\CFH}{N-1}\int_{\R^{Nd}}\left(\sum_{1\leq i<j\leq N}\frac{1}{|x_i-x_j|^{2s}}\right)|u(x)|^2\,dx \\  
			&\hspace{-3em}+\ \frac{2\,\CFH}{(N-1)(N-2)}\int_{\R^{Nd}}\remR(x)|u(x)|^2\,dx.
		\end{align*}
		\endgroup
		By Proposition~\ref{prop:uniform:domination},
		\[
		\remR(x) \ \geq\ \delta\sum_{1\leq i<j\leq N}\frac{1}{|x_i-x_j|^{2s}}\quad\text{ for all }x\in\off{N}.
		\]
		Substituting this bound into the preceding inequality yields the claimed estimate. 
	\end{proof}
	
	\begin{remark}
		The above corollary is an analogue of Theorem 5 in Lundholm~\cite{lundholm2017methods} in the setting of many-particle fractional inequalities with Coulomb-type potentials. It is an effective improvement of the basic inequality~\eqref{eq:intro:fracHardyQuadraticVersion} but, unlike the local case, we have no explicit constant. In addition, the improvement requires the restriction $N>d+1$. Further comments are made below in Remark~\ref{rem:behavior:s}.
	\end{remark}    
	\medskip
	
	\subsection{Final remarks}\label{sec:final-remarks}
	
	We conclude with remarks concerning limitations of the method and possible directions for improvement. 
	
	\begin{remark}\label{rem:new1}
		The asymptotic behaviour of constants in fractional Hardy inequalities with Coulomb-type potentials has been studied in several works, most recently by Frank, Hoffmann-Ostenhof, Laptev, and Solovej~\cite{frank2024hardy}. Our result does not constitute an asymptotic result like theirs, and the method of proof differs substantially from theirs. While our approach is based on the ground-state representation combined with explicit singular integral identities, their analysis relies on semiclassical methods, coherent states, mean-field arguments, and the Lieb-Thirring inequality, and is closely connected to the Thomas-Fermi theory for Coulomb systems. 
	\end{remark}
	
	\begin{remark}[\textbf{Behaviour as $s\uparrow 1$}]\label{rem:behavior:s} We highlight some inherent difficulties in obtaining sharp constants for the Coulomb-type potential $\sum_{i<j}|x_i-x_j|^{-2s}$ in the fractional setting. This is best illustrated by considering the limit $s\uparrow 1$. We focus on the three-particle case, which captures the essential difficulty in both local and nonlocal settings. On the one hand, as $s\uparrow 1$, the fractional three-particle potential $\V(x_1,x_2,x_3)$ converges to the local three-body potential 
		\[
		V_{3,loc}(x_1,x_2,x_3) \ :=\ \sum_{\text{cyclic}}\frac{|x_k-x_j|^2}{|x_i-x_j|^2|x_i-x_k|^2}
		\ =\ 2\;\sum_{1\leq i<j\leq 3}\frac{1}{|x_i-x_j|^2}-\frac{1}{R^2_{123}},
		\]
		while the coefficient $\CFH$ in~\eqref{eq:thm:1:inequality} converges to the optimal constant $(d-2)^2/4$ in the two-particle Hardy inequality.
		
		On the other hand, we compare this limit with the local result in~\cite[Theorem~5]{lundholm2015geometric}. There, the ground state is taken in the form
		\[
		|x_1-x_2|^{-(d-2)\alpha}|x_1-x_3|^{-(d-2)\alpha}|x_2-x_3|^{-(d-2)\alpha}\,,
		\]
		which leads to the lower bound $(d-2)^2\bigl(2\alpha(1-\alpha)\sum_{1\leq i<j\leq 3}|x_i-x_j|^{-2}-(\alpha^2/R^2_{123})\bigr)$. This expression can still be optimised in $\alpha$; the optimal choice $\alpha=1/3$ yields the improved constant $(d-2)^2/3$ in the work of Lundholm~\cite{lundholm2015geometric}, compared to $(d-2)^2/4$ for the Coulomb-type potential. We cannot match this result in the limiting case of our analysis as $s \to 1$, because the minimum of $V_{3,loc}(x_1,x_2,x_3)/(r^{-2}_{12}+r^{-2}_{13}+r^{-2}_{23})$ over $(x_1,x_2,x_3)\in\off{3}$ is~$1$. Indeed, the condition $N=3>d+1$ does not hold, and Remark~\ref{rem:necesity:condition} shows that no positive $\delta$ is available in Proposition~\ref{prop:uniform:domination}. Consequently, this limiting argument does not improve the constant $(d-2)^2/4$.
		
		This identifies the core difficulty in obtaining an improved constant for the Coulomb-type potential in the nonlocal setting. Although in the evaluation of the interaction term in Lemma~\ref{lem:limit:interaction} we start with a ground-state ansatz containing a free parameter $\alpha$,
		\[
		|x_i-x_j|^{-(d-2)\alpha}|x_i-x_k|^{-(d-2)\alpha}|x_j-x_k|^{-(d-2)\alpha}\,,
		\]
		Theorem~\ref{thm:main:2} can only be applied when $\alpha=1/2$.
	\end{remark}
	
	\begin{remark}[\textbf{Limitations of the GSR method in the fractional setting}] If one attempts to extend the ground-state representation method to four particles, the natural ansatz for the ground state would be
		\[
		\omega(x_1,x_2,x_3,x_4) \ =\ \prod_{1\leq i<j\leq 4}|x_i-x_j|^{-\alpha}.
		\]
		This leads to a four-fold singular integral of the form	
		\[
		\intd \prod^4_{i=1}|x_i-t|^{-\beta_i}\,dt
		\]
		with suitable exponents $\beta_i$. However, unlike the two-fold and three-fold cases, analogous formulae for such $k$-fold integrals, for $k\geq 4$, are unavailable in general, according to recent work by Wu, Shi, Nie, and Yan~\emph{\cite{wu2020kfold}}.
	\end{remark}

	\begin{remark}[{\bf On the dimensional threshold $d\geq 4-2s$}]\label{rem:dim:restriction}
		The restriction $d\geq 4-2s$ arises from the specific mechanism used to evaluate the interaction term in Lemma~\ref{lem:limit:interaction}. More precisely, after reducing the problem to the model integral $T(x_1,x_2,x_3)$ and performing the Green-type decomposition, one encounters boundary contributions whose behaviour changes qualitatively at $d+2s=4$. Above the threshold these terms are negligible, whereas in the borderline case they remain finite and can still be computed explicitly; below it, control is lost over these terms. We emphasise that this does not itself show that the inequality of Theorem~\ref{thm:main:1} is false when $d+2s<4$. Rather, it shows that the present combination of truncated ground-state representation, fractional Leibniz rule, and explicit Selberg-type correlation identity appears insufficient in that range, although it remains plausible that some version of Theorem~\ref{thm:main:1} holds below this threshold.
	\end{remark}
	\section{Proof of Theorem~\ref{thm:main:2}}\label{subsec:new_integral_formula}
	
	We prove the singular integral identity~\eqref{eq:thm:MOZ:1} in three steps. First, we exploit the invariances of the kernel to reduce the statement to a normalised configuration. Second, we show that the resulting objects define tempered distributions. Third, we compute their Fourier transforms and compare them.\medskip
	
	\noindent\textbf{Step 1: Reduction by invariance.}\\
	Let 
	\begin{equation}\label{eq:def:K}
		\kernel(x,y,z) \ :=\ \intd\frac{(x-t)\cdot(y-t)}{|x-t|^{d_2+1}|y-t|^{d_3+1}}\dfrac{1}{|z-t|^{d_1}}dt\,.
	\end{equation}
	We suppress the dependence on the parameters $d_1$, $d_2$ and $d_3$ in the notation. The assumptions $d_i<d$ for $i=1,2,3$ ensure local integrability of the integrand, while the constraint 
	\[
	d_1+d_2+d_3=2d
	\] 
	implies integrability at infinity, since the integrand behaves like $|t|^{-2d}$ as $|t|\to\infty$.
	
	The kernel $\kernel$ is invariant under simultaneous translations and rotations, and is homogeneous of degree $-d$. More precisely, for all $a,x,y,z \in \R^d$, all $Q\in SO(\R^d)$ and all $\rho>0$, one has   
	\begin{gather}
		\label{I1} \kernel(x+a,y+a,z+a)\ =\ \kernel(x,y,z)\,, \\
		\label{I2} \kernel(Qx,Qy,Qz) \ =\ \kernel(x,y,z)\,,\\
		\label{I3} \kernel(\rho x,\rho y,\rho z)\ =\ \rho ^{-d} \kernel(x,y,z).
	\end{gather}
	Likewise, define
	\begin{equation}\label{J}
		J(x,y,z) \ :=\ \frac{(x-z)}{|x-z|^{d-d_3+1}} \cdot \frac{(y-z)}{|y-z|^{d-d_2+1}}\frac{1}{|x-y|^{d-d_1}}\,.
	\end{equation} 
	The function $J$ satisfies the same transformation laws, and the homogeneity property in~\eqref{I3} follows directly from the relation $d_1+d_2+d_3=2d$. Therefore, it suffices to establish 
	\begin{equation}\label{reduced_identity} 
		\kernel(w,e_1,0)\ =\ C\, J(w,e_1,0), \quad\mbox{ whenever } w \in \R^d \setminus \{0,e_1\},
	\end{equation}
	for a suitable constant $ C=\CMOZ$. From now on, we  write simply $\kernel(w)$ and $J(w)$ for $\kernel(w,e_1,0)$ and $J(w,e_1,0)$, respectively.\medskip
	
	\noindent\textbf{ Step 2: Tempered-distribution framework.}\\
	We first verify that both $\kernel$ and $J$ define tempered distributions. In the normalised configuration, these functions take the form
	\begin{gather}
		\kernel(w) \ =\ \intd\frac{(w-t)\cdot(e_1-t)}{|w-t|^{d_2+1}|e_1-t|^{d_3+1}|t|^{d_1}}\,dt\,,\label{Fw} \\[0.5em]
		J(w) \ =\ \frac{w_1}{|w-e_1|^{d-d_1}|w|^{d-d_3+1}} \,.\label{Jw}
	\end{gather}
	The function $J$ is locally integrable, since $d_1,d_3>0$, and decays at infinity like $|w|^{-d_2}$. Local integrability together with polynomial decay at infinity ensures that $J$ defines a tempered distribution. 
	
	For $\kernel$, we claim that it has a singularity of order $|w-e_1|^{-(d-d_1)}$ at $e_1$, a singularity of order $|w|^{-(d-d_3)}$ at $0$, and decays like $|w|^{-d_2}$ at infinity. To analyse the behaviour near~$e_1$, write $w=e_1+ru$ with $u\in\S^{d-1}$, and perform the change of variables $t=e_1+r s$ in~\eqref{Fw}. This yields
	\begin{equation*}
		\kernel(e_1+ru)
		\ =\ r^{d-d_2-d_3} \intd\frac{(u-s)\cdot(-s)}{|u-s|^{d_2+1}\,|s|^{d_3+1}\,|e_1+rs|^{d_1}}\,ds\,.
	\end{equation*}
	As $r\to0$, the integral converges to 
	\begin{equation*}
		\intd\frac{(u-s)\cdot(-s)}{|u-s|^{d_2+1}|-s|^{d_3+1}}\, ds\,,
	\end{equation*}
	a finite quantity independent of $u\in\S^{d-1}$ by rotational symmetry. Since $d_1+d_2+d_3=2d$, it follows that 
	\[
	\kernel(e_1+ru)=O\left(r^{-(d-d_1)}\right).
	\]    
	The behaviour near $0$ and at infinity follows analogously. In particular, $\kernel$ is locally integrable and has at most polynomial growth at infinity, and hence defines a tempered distribution.\medskip
	
	\noindent\textbf{Step 3: Comparison of Fourier transforms.}\\
	We prove identity~\eqref{reduced_identity} in the sense of tempered distributions by showing that the Fourier transforms of $\kernel$ and $J$ coincide up to an explicit constant.
	
	Let $\varphi\in\mathcal{S}(\R^d)$. Identifying $\kernel$ with the tempered distribution induced by the kernel in~\eqref{Fw}, we write 
	\begin{equation*}
		\langle\widehat{\kernel} \,, \, \varphi \rangle 
		\ =\ \langle\kernel \,, \, \widehat{\varphi}\,\rangle \ =\ \intd \left( \intd\frac{(w-t)\cdot(e_1-t)}{|w-t|^{d_2+1}|e_1-t|^{d_3+1}|t|^{d_1}}\, dt \right) \widehat{\varphi}(w)\, dw
	\end{equation*}
	and decompose the dot product componentwise,
	\begin{equation}\label{Igh}
		\langle\widehat{\kernel} \,, \, \varphi \rangle
		\ =:\ \sum_{j=1}^d  \intd \left(\intd g_j(w-t) h_j(t)\, dt \right) \widehat{\varphi}(w)\, dw\,,
	\end{equation}
	where 
	\begin{equation}\label{eq:def:gj:hj} 
		g_j(t) \ :=\ \dfrac{t_j}{|t|^{d_2+1}}\,,\quad
		h_j(t) \ :=\ \dfrac{(e_1-t)_j}{|e_1-t|^{d_3+1}|t|^{d_1}}\,.
	\end{equation}
	Here $g_j$ defines a tempered distribution, while $h_j\in L^1(\R^d)$. Using first the Grafakos-Morpurgo formula~\eqref{eq:prelim:grafakos:formula} and then the finiteness of the final integral, one obtains the following bound: 
	\begin{align*}
		&\intd \intd |g_j(w-t)\,h_j(t)|\, |\widehat{\varphi}(w)|\, dt\, dw \\
		& \quad\ \le\  \intd \intd \dfrac{1}{|w-t|^{d_2}} \dfrac{1}{|e_1-t|^{d_3}}\dfrac{1}{|t|^{d_1}} \ |\widehat{\varphi}(w)|\, dt\, dw\\
		& \quad\ =\ C_{\rm GM}(d_1,d_2,d_3,d)\intd  \dfrac{1}{|w-e_1|^{d-d_1}}\dfrac{1}{|w|^{d-d_3}} \ |\widehat{\varphi}(w)|\, dw\,,
	\end{align*}
	This proves the integrability of  $g_j(w-t) h_j(t)\,   \widehat{\varphi}(w)$ and,
	therefore, the order of integration in~\eqref{Igh} can be interchanged. Using standard identities for Fourier transforms of translates and convolutions, we obtain
	\begingroup
	\allowdisplaybreaks
	\begin{align*}
		\langle \widehat{\kernel} \,, \, \varphi \rangle & \ =\ \sum_{j=1}^d  \intd \left(\intd g_j(w-t)  \widehat{\varphi}(w)\, dw \right)h_j(t)\, dt \\
		&\ =\ \sum_{j=1}^d  \intd \left(\intd \wwidehat{g_j(\cdot-t)}(x)\varphi(x)\, dx  \right)h_j(t)\, dt\\
		& \ =\ \sum_{j=1}^d  \intd \left(\intd e^{-2 \pi i t \cdot x}\widehat{\,g_j\,}(x)  \varphi(x)\,dx\right)h_j(t)\, dt\\
		& \ =\ \sum_{j=1}^d  \intd \left(\intd e^{-2 \pi i t \cdot x}h_j(t)\, dt  \right)\widehat{\,g_j\,}(x)\varphi(x)\, dx \\
		&\ =\ \sum_{j=1}^d\intd\wwidehat{h_j}(x) \widehat{\,g_j\,}(x)\varphi(x)\, dx\,.
	\end{align*}
	\endgroup
	In the penultimate step, we again interchanged the order of integration, justified by the integrability of $h_j(t)\,\widehat{\,g_j\,}(x)\,\varphi(x)$ on $\R^d\times\R^d$. Hence,
	\begin{equation}\label{ft-I}
		\widehat{\kernel}\ =\ \sum_{j=1}^d\wwidehat{h_j}\,\widehat{\,g_j\,}.
	\end{equation}
	It remains to identify $\widehat{J}$. This is precisely the content of the following claim, proved later in Proposition~\ref{prop:app:Fourier:convolution}. 
	
	\medskip
	\noindent
	{\bf Claim. } 
	\begingroup
	\allowdisplaybreaks
	\begin{align}\label{FT-product}	 
		\left(\dfrac{w_1}{|w|^{d-d_3+1}} \dfrac{1}{|w-e_1|^{d-d_1}}\right)^{\wwidehat{\;\;\;}}(\xi)
		\ =\  \left(\frac{\CLL{d_1}}{\CLL{d-d_1}}\frac{\CLL{d_3+1}}{i\,\CLL{d-d_3+1}}\right) \left( \frac{y_1}{|y|^{d_3+1}} \ast \frac{e^{-2 \pi i y \cdot e_1}}{|y|^{d_1}} \right)(\xi).
	\end{align}
	\endgroup
	\smallskip
	
	\noindent From this point on, assume that~\eqref{FT-product} holds and first write explicitly the convolution on the right-hand side:
	\[
	\int_{\R^d}e^{-2\pi iy\cdot e_1}\frac{(\xi-y)\cdot e_1}{|\xi-y|^{d_3+1}}\frac{1}{|y|^{d_1}}\,dy.
	\]
	Apply the change of variables $y=|\xi|s$ in the convolution integral to get
	\begin{equation*}
		\biggl(i \frac{\CLL{d-d_1}}{\CLL{d_1}}\frac{\CLL{d-d_3+1}}{\CLL{d_3+1}}\biggr)\widehat{J} (\xi) 
		\ =\ \frac{1}{|\xi|^{d_1+d_3-d}}\intd e^{-2 \pi i |\xi|s\cdot e_1}\frac{(\widehat{\xi}-s) \cdot e_1}{|\widehat{\xi}-s|^{d_3+1}} \frac{1}{|s|^{d_1}}\,ds,
	\end{equation*}
	where $\widehat{\xi}$ denotes the unit vector $\xi/|\xi|\in\S^{d-1}$. Since the spatial dimension is at least $2$, one can find a {\em symmetric} orthogonal change of coordinates $s \mapsto Qz$ with $\widehat{\xi} = Qe_1$. Hence,
	\begingroup
	\allowdisplaybreaks
	\begin{eqnarray*}
		\biggl(i\frac{\CLL{d-d_1}}{\CLL{d_1}}\frac{\CLL{d-d_3+1}}{\CLL{d_3+1}}\biggr)
		\widehat{J} (\xi) 
		& = & \frac{1}{|\xi|^{d_1+d_3-d}}\intd e^{-2 \pi i |\xi|(Qz)\cdot e_1}\frac{(Qe_1- Qz) \cdot Q\widehat{\xi} }{|Qe_1-Qz|^{d_3+1}} \frac{1}{|Qz|^{d_1}}\,|\det Q| dz\\
		& = &\ \sum_{j=1}^d\dfrac{\xi_j}{|\xi|^{d-d_2+1}} \intd  e^{-2 \pi i z\cdot\xi}\frac{(e_1-z)_j}{|e_1-z|^{d_3+1}} \frac{1}{|z|^{d_1}}\, dz.
	\end{eqnarray*}
	\endgroup
	Invoking \eqref{lem:app:Fourier} once again, for every $j\in\{1,\ldots, d\}$, we have
	\[
	\frac{\xi_j}{|\xi|^{d-d_2+1}}\ =\ \left(\frac{\CLL{d-d_2+1}}{i\,\CLL{d_2+1}}\right)^{-1}\widehat{\frac{x_j}{|x|^{d_2+1}}}(\xi)\,.
	\]
	Substituting this into the previous relation gives
	\begin{equation}\label{ft-J}
		\biggl(\frac{\CLL{d-d_1}}{\CLL{d_1}}\frac{\CLL{d-d_2+1}}{\CLL{d_2+1}}\frac{\CLL{d-d_3+1}}{\CLL{d_3+1}}\biggr)\widehat{J} (\xi)
		\ =\ \sum_{j=1}^d \widehat{\,g_j\,}(\xi) \wwidehat{h_j}(\xi),
	\end{equation}
	in view of~\eqref{eq:def:gj:hj}. Comparing expressions~\eqref{ft-I} and~\eqref{ft-J}, we conclude that
	\begin{equation}\label{eq:proof:theorem:MOZ:last}
		\kernel(w) \ =\ \left(\frac{\CLL{d-d_1}}{\CLL{d_1}}\frac{\CLL{d-d_2+1}}{\CLL{d_2+1}}\frac{\CLL{d-d_3+1}}{\CLL{d_3+1}}\right) J(w).
	\end{equation}
	A direct simplification of the constants using the definition of $\CLL{\mu}$ in~\eqref{eq:Fourier:Lieb-Loss:constant} yields precisely the constant~$\CMOZ$,
	\begin{align*}
		\frac{\CLL{d-d_1}}{\CLL{d_1}}\frac{\CLL{d-d_2+1}}{\CLL{d_2+1}}\frac{\CLL{d-d_3+1}}{\CLL{d_3+1}}
		&\ =\ \pi^{\frac{d}{2}}\,\frac{\Gamma(\tfrac{d-d_1}{2})}{\Gamma(\tfrac{d_1}{2})}\frac{\Gamma(\tfrac{d-d_2+1}{2})}{\Gamma(\tfrac{d_2+1}{2})}\frac{\Gamma(\tfrac{d-d_3+1}{2})}{\Gamma(\tfrac{d_3+1}{2})}.
	\end{align*}
	This completes the proof of Theorem~\ref{thm:main:2}.\qed 	
	\medskip
	
	\appendix
	
	\makeatletter
	\renewcommand{\@seccntformat}[1]{%
		\ifcsname ifinappendix\endcsname
		\ifinappendix
		\ifdefstring{#1}{section}{\appendixname~\thesection.\quad}{%
			\ifdefstring{#1}{subsection}{\appendixname~\thesubsection.\quad}{%
				\csname the#1\endcsname\quad
			}%
		}%
		\else
		\csname the#1\endcsname\quad
		\fi
		\else
		\csname the#1\endcsname\quad
		\fi
	}
	\makeatother
	
	\section{Auxiliary analytic results}\label{sec:app:A}
	This appendix contains the auxiliary results used in the proofs of Theorems~\ref{thm:main:1} and~\ref{thm:main:2}. Appendix~\ref{subsec:app:A:1} treats the truncated ground-state construction and the limiting two-body and interaction terms in the proof of Theorem~\ref{thm:main:1}. Appendix~\ref{subsec:app:A:2} proves the distributional convolution identity used in Theorem~\ref{thm:main:2}. In Appendix~\ref{subsec:app:A:3} we record the convolution estimates used in the proof of Theorem~\ref{thm:main:2}. 
	
	\subsection{Technical lemmas for the proof of Theorem~\ref{thm:main:1} }\label{subsec:app:A:1}
	
	We prove the auxiliary results required for the ground-state representation arguments in Section~\ref{sec:proof:main_thm}: admissibility of the truncated weights, convergence of the two-body Hardy terms, and evaluation of the nonlocal interaction term.
	
	\begin{lemma}\label{lem:app:regularity}
		Let $d\in\N$, $0<s<1$, $0<\alpha<\frac{d-2s}{2}$, and $n\in\N$. Then, for every $i\in\{1,2,3\}$,
		\[
		\om{i} \in H^s(\R^d),
		\]
		where $\om{i}$ was introduced in Step 1 of the proof of Theorem~\ref{thm:main:1}.
	\end{lemma}
	
	\begin{proof}[Proof of Lemma~\ref{lem:app:regularity}]
		By symmetry it suffices to consider the case $i=1$. Fix distinct points $x_2,x_3\in\R^d$, and regard $\om{1}$ as a function of the variable $x_1\in\R^d$. Recall that
		\[
		v_{\alpha,n}(x) \ :=\ |x|^{-\alpha}\phi\left(\frac{x}{n}\right)\quad\text{ for }x\in\R^d,
		\]
		where $\phi\in C^{\infty}_c(\R^d)$ is a radial function satisfying $0\leq \phi\leq 1$ and $\phi\equiv 1$ for $|x|\leq 1$, and that  
		\[
		\om{1}(x_1) \ =\ |x_2-x_3|^{-\alpha}\,\valphan(x_1-x_2)\,\valphan(x_1-x_3).
		\]
		By Frank-Lieb-Seiringer~\cite[Proposition 4.1]{frank2008hardy}, one has
		\[
		\valphan\in H^s(\R^d) \quad\text{ for every }n\in\N,
		\]
		provided $0<\alpha<(d-2s)/2$. Since $H^s(\R^d)$ is translation invariant, also
		\[
		\valphan(\,\cdot-x_2),\, \valphan(\,\cdot-x_3)\in H^s(\R^d).
		\]
		It remains to justify that their product belongs to $H^s(\R^d)$. Set
		\[
		r\ :=\ \frac{1}{4}|x_2-x_3|, 
		\]
		and choose a partition of unity $\eta_2,\eta_3,\eta_{\infty}$. That is, $\eta_2,\eta_3\in C^{\infty}_c(\R^d)$ with $\eta_2$ supported near $x_2$, $\eta_3$ supported near $x_3$, and $\eta_{\infty} :=1-\eta_2-\eta_3$. Hence
		\[
		\eta_{\infty} \ +\ \eta_2\ +\ \eta_3\\ =\ 1,
		\]
		with
		\[
		\eta_2 \ \equiv 1\ \;\,\text{ on }B_r(x_2),\quad \eta_3\ \equiv\ 1 \;\,\text{ on }B_r(x_3),
		\]
		and
		\[
		\supp \eta_2 \ \subset\ B_{2r}(x_2),\qquad \supp \eta_3 \ \subset\ B_{2r}(x_3).
		\]
		Because $B_{2r}(x_2)\, \cap\, B_{2r}(x_3)=\emptyset$, the factor $\valphan(\,\cdot-x_3)$ is smooth up to $\supp \eta_2$, and $\valphan(\,\cdot-x_2)$ is smooth on $\supp \eta_3$. Hence,
		\[
		\eta_2\,\om{1} \ =\ m_2\,\valphan(\;\cdot-x_2),\qquad m_2(x_1) \:=\ |x_2-x_3|^{-\alpha}\,\eta_2(x_1)\,\valphan(x_1-x_3),
		\]
		and $m_2\in C^{\infty}_c(\R^d)$. Since multiplication by a $C^{\infty}_c$ function is bounded on $H^s(\R^d)$,
		\[
		\eta_2\,\om{1}\in H^s(\R^d).
		\]
		Applying the same reasoning gives
		\[
		\eta_3\,\om{1}\in H^s(\R^d).
		\]
		Finally, on $\supp\eta_{\infty}$ both factors $\valphan(\;\cdot-x_2)$ and $\valphan(\;\cdot-x_3)$ are smooth. Since each $\valphan$ is compactly supported, their product is smooth and compactly supported. Therefore,
		\[
		\eta_{\infty}\,\om{1}\in C^{\infty}_c(\R^d)\subset H^s(\R^d).
		\]
		Summing the three parts yields
		\[
		\om{1} \ = \ \eta_2\,\om{1}\ +\ \eta_3\,\om{1}\ +\ \eta_{\infty}\,\om{1}\in H^s(\R^d).
		\]
		This proves the lemma.
	\end{proof}
	\noindent
	
	\begin{lemma}\label{lem:app:LapIPP:2}
		Let $s\in(0,1)$ be arbitrary. For any $u\in C^{\infty}_c(\off{3})$ and $i,j\in\{1,2,3\}$ with $i\neq j$, there holds
		\begin{align}\label{eq:lem:app:LapIPP:2}
			\lim_{\alpha\uparrow\frac{d-2s}{2}}\,\lim_{n\to\infty}
			&\inn{\sLapj{i}\valphan(x_i-x_j)}{\frac{|u(x_1,x_2,x_3)|^2}{\valphan(x_i-x_j)}}\notag\\
			& \hspace{6em}\ =\  \CFH \intd\frac{|u(x_1,x_2,x_3)|^2}{|x_i-x_j|^{2s}}\,dx_i\,.
		\end{align}
	\end{lemma}
	\begin{proof}[Proof of Lemma~\ref{lem:app:LapIPP:2}] 
		We argue only the case $i=1$ and $j=2$, since the remaining cases follow by relabelling. Fix $x_2\neq x_3$. Since $u\in C^{\infty}_c(\off{3})$, the function
		\[
		x_1\mapsto u(x_1,x_2,x_3)
		\]
		has compact support in $\R^d\setminus\{x_2\}$. Hence there exists $n_0\in\N$ such that for every $n\geq n_0$,
		\[
		\supp u(\,\cdot,x_2,x_3) \ \subset\ B_n(x_2).
		\]
		Since $\phi\equiv 1$ on $B_1(0)$, it follows that 		for any $n\geq n_0$ large,
		\[
		\valphan(x_1-x_2) \ =\ |x_1-x_2|^{-\alpha}\quad\text{ for }\;x_1\in\supp u(\,\cdot,x_2,x_3).
		\] 
		Set
		\[
		\psi_{\alpha}(x_1)\ :=\ |u(x_1,x_2,x_3)|^2\,|x_1-x_2|^{\alpha}\in C^{\infty}_c(\R^d),
		\]
		and observe that, for $n\geq n_0$, the quotient $|u(\cdot, x_2,x_3)|^2/\valphan(\,\cdot-x_2)$ is represented by $\psi_{\alpha}$. In particular,
		\[
		\inn{\sLapj{1}\valphan(\,\cdot-x_2)}{\frac{|u(\cdot, x_2,x_3)|^2}{\valphan(\,\cdot-x_2)}} \ =\ \inn{\sLapj{1}\valphan(\,\cdot-x_2)}{\psi_{\alpha}}
		\]
		for $n\geq n_0$. We now argue how to pass to the limit $n\to\infty$. First, by duality
		\[
		\inn{\sLapj{1}\valphan(\,\cdot-x_2)}{\psi_{\alpha}} \ =\ \inn{\valphan(\,\cdot-x_2)}{\sLapj{1}\psi_{\alpha}}.
		\]
		Now, since $\psi_{\alpha}\in C^{\infty}_c(\R^d)$, the function $\sLapj{1}\psi_{\alpha}$ is smooth and decays as follows
		\[
		\left|\sLapj{1}\psi_{\alpha}(y)\right|\ \leq\ C(1+|y|)^{-d-2s}.
		\]
		Moreover, observe that
		\[
		0 \ \leq\ \valphan(y-x_2) \ \leq\ |y-x_2|^{-\alpha}\;\;\text{ and }\;\; \valphan(y-x_2)\xrightarrow[n\to\infty]{}|y-x_2|^{-\alpha},\;\quad\text{for a.e. }y\in\R^d.
		\]
		Finally, the function $|y-x_2|^{-\alpha}|\sLapj{1}\psi_{\alpha}(y)|\in L^1(\R^d)$ is integrable: near $y=x_2$ this follows from $\alpha<d$, and integrability at infinity follows from the decay $|y|^{-d-2s-\alpha}$.
		The dominated convergence theorem then yields
		\[
		\lim_{n\to\infty}\inn{\valphan(\,\cdot-x_2)}{\sLapj{1}\psi_{\alpha}} \ =\ \inn{\;|\,\cdot-x_2|^{-\alpha}}{\sLapj{1}\psi_{\alpha}}.
		\]
		That is to say,
		\[
		\lim_{n\to\infty}\inn{\sLapj{1}\valphan(\,\cdot-x_2)}{\psi_{\alpha}} \ =\ \inn{\sLapj{1}|\,\cdot-x_2|^{-\alpha}}{\psi_{\alpha}}.
		\]
		Lastly, the power-law identity~\eqref{fracLap-Rszpot} gives
		\[
		\lim_{\alpha\uparrow\frac{d-2s}{2}}\inn{\sLapj{1}|\,\cdot-x_2|^{-\alpha}}{\psi_{\alpha}}\ =\ \CFH\,\intd\frac{|u(x_1,x_2,x_3)|^2}{|x_1-x_2|^{2s}}\,dx_1,
		\]
		and the proof is now complete.
	\end{proof}	    
	\begin{lemma}\label{lem:limit:interaction} 
		Let $s\in(0,1)$ and $d\geq 4-2s$. Given  distinct points $x_i,x_j,x_k\in\R^{d}$, one has
		\begingroup
		\allowdisplaybreaks
		\begin{align}\label{cross-term-limit}
			&\lim_{\alpha\,\uparrow\frac{d-2s}{2}}\,\lim_{n\to\infty} \frac{\Nn(\valphan(\,\cdot-x_j), \valphan(\,\cdot-x_k))(x_i)}{\valphan(x_i-x_j)\,\valphan(x_i-x_k)} \notag\\
			&\ =\  \CFH \left(\frac{1}{|x_i-x_j|^{2s}}+\frac{1}{|x_i-x_k|^{2s}}+2\,\frac{(x_k-x_i)}{|x_k-x_i|^{2s}}\cdot\frac{(x_j-x_i)}{|x_j-x_i|^{2s}}\,\frac{1}{|x_j-x_k|^{2-2s}}\right.\notag \\[.5em]
			&\hspace{7em}\left. - \left[\frac{|x_i-x_k|^{2-2s}}{|x_i-x_j|^{2s}|x_j-x_k|^{2-2s}}+\frac{|x_i-x_j|^{2-2s}}{|x_i-x_k|^{2s}|x_j-x_k|^{2-2s}}\right]\right) .
		\end{align} 
		\endgroup
	\end{lemma}    
	
	\begin{proof}[Proof of Lemma~\ref{lem:limit:interaction}]
		We treat only the case $i=1$, $j=2$ and $k=3$ since the remaining cases follow by relabelling. It suffices to identify the limiting interaction term generated by the fractional Leibniz rule. The proof is divided into five steps.
		
		\medskip 
		\noindent\textbf{Step 1: Reduction to the model integral and passing to the limit.}\\
		For $0<\alpha\leq (d-2s)/2$ and $n\in\N$, define
		\[
		F_{\alpha,n}(y) \ :=\ \frac{\bigl(\valphan(x_1-x_2)-\valphan(y-x_2)\bigr)\bigl(\valphan(x_1-x_3)-\valphan(y-x_3)\bigr)}{\valphan(x_1-x_2)\,\valphan(x_1-x_3)\,|x_1-y|^{d+2s}}
		\]
		in such a way that
		\[
		\frac{\Nn\bigl(\valphan(\;\cdot-x_2),\valphan(\;\cdot-x_3)\bigr)(x_1)}{\valphan(x_1-x_2)\,\valphan(x_1-x_3)}
		\ =\ 2a_{s,d}\intd F_{\alpha,n}(y)\,dy.
		\]
		Choose $n_0\in\N$ so large that
		\[
		|x_1-x_2|<n\quad\text{ and }\quad |x_1-x_3|<n
		\]
		for the fixed $x_1$ and $n\geq n_0$. Now recall the definition $\valphan$ in formula~\eqref{eq:proof:main:1:GSR:1}. Since $\phi\equiv 1$ on $B_1(0)$, one has
		\[
		\valphan(x_1-x_2) \ =\ |x_1-x_2|^{-\alpha}\quad\text{ and }\quad \valphan(x_1-x_3) \ =\ |x_1-x_3|^{-\alpha}\,
		\]
		for all $x_1 \in \supp u(\cdot,x_2,x_3)$ and every $n\geq n_0$. Hence,
		\[
		F_{\alpha,n}(y) \ =\ \left(1-\frac{|x_1-x_2|^{\alpha}}{|y-x_2|^{\alpha}}\,\phi\left(\frac{y-x_2}{n}\right)\right)\left(1-\frac{|x_1-x_3|^{\alpha}}{|y-x_3|^{\alpha}}\,\phi\left(\frac{y-x_3}{n}\right)\right)\,|x_1-y|^{-d-2s}.
		\]
		Thus, for each fixed $y\notin\{x_1,x_2,x_3\}$, one has the convergence
		\begin{equation}\label{convFalphan}
			F_{\alpha,n}(y) \longrightarrow \left(1-\frac{|x_1-x_2|^{\frac{d-2s}{2}}}{|y-x_2|^{\frac{d-2s}{2}}}\right)\left(1-\frac{|x_1-x_3|^{\frac{d-2s}{2}}}{|y-x_3|^{\frac{d-2s}{2}}}\right)\,|x_1-y|^{-d-2s}
		\end{equation}
		by taking the limit first as $n\to\infty$ and then as $\alpha\uparrow(d-2s)/2$. Now, set
		\[
		r \ :=\ \frac{1}{4}\min\{|x_1-x_2|,|x_1-x_3|,|x_2-x_3|\}. 
		\]
		We study the behaviour of $F_{\alpha,n}$ separately on $B_r(x_1)$, $B_r(x_2)$, $B_r(x_3)$, and at infinity. Near $x_1$, both factors in parenthesis vanish at $y=x_1$. Since the functions
		\[
		y\mapsto|y-x_2|^{-\alpha}\quad\text{ and }\quad y\mapsto|y-x_3|^{-\alpha}
		\]
		are smooth on $B_r(x_1)$, the mean-value theorem gives
		\[
		\left|1-\frac{|x_1-x_2|^{\alpha}}{|y-x_2|^{\alpha}}\,\phi\left(\frac{y-x_2}{n}\right)\right|+\left|1-\frac{|x_1-x_3|^{\alpha}}{|y-x_3|^{\alpha}}\,\phi\left(\frac{y-x_3}{n}\right)\right|
		\ \leq\ C|y-x_1|.
		\]
		Therefore,
		\[
		|F_{\alpha,n}(y)|
		\ \leq\ C|y-x_1|^{2-d-2s}\quad\text{ for }\;y\in B_r(x_1).
		\]
		This bound is locally integrable near $x_1$ since $s<1$. Near $x_2$, the second factor is bounded on $B_r(x_2)$, while
		\[
		\left|1-\frac{|x_1-x_2|^{\alpha}}{|y-x_2|^{\alpha}}\,\phi\left(\frac{y-x_2}{n}\right)\right| 
		\ \leq\ C\Bigl(1+|y-x_2|^{-\frac{d-2s}{2}}\Bigr).
		\]
		Moreover, the lower bound $|x_1-y|\geq 3r$ on $B_r(x_2)$ implies that
		\[
		|F_{\alpha,n}(y)| \ \leq\ C\Bigl(1+|y-x_2|^{-\frac{d-2s}{2}}\Bigr)\quad\text{ for }\; y\in B_r(x_2),
		\]
		which is integrable because $(d-2s)/2<d$. The same argument applies near $x_3$.
		
		\noindent
		Finally, the behaviour of $F_{\alpha,n}(y)$ at infinity is like $|y|^{-d-2s}$ 
		which is integrable as $s>0$. Thus, the family $\{F_{\alpha,n}: \alpha < (d-2s)/2\,,\, n \in \mathbb{N}\}$ is dominated by an $L^1(\R^d)$-function. Hence, the dominated convergence theorem gives
		\begin{equation}\label{interaction-term}
			\lim_{\alpha\uparrow\frac{d-2s}{2}}\lim_{n\to\infty}\frac{\Nn(\valphan(\;\cdot-x_2),\valphan(\;\cdot-x_3))(x_1)}{\valphan(x_1-x_2)\,\valphan(x_1-x_3)} \ =\ 2a_{s,d}\, T(x_1,x_2,x_3),
		\end{equation}
		with the right-hand side defined as
		\[
		T(x_1,x_2,x_3) \ :=\ \intd \left(1-\frac{|x_1-x_2|^{\frac{d-2s}{2}}}{|y-x_2|^{\frac{d-2s}{2}}}\right)\left(1-\frac{|x_1-x_3|^{\frac{d-2s}{2}}}{|y-x_3|^{\frac{d-2s}{2}}}\right)\,|x_1-y|^{-d-2s}\,dy.
		\]
		\smallskip
		
		\noindent\textbf{Step 2: Reduction by integration by parts.}\\
		We now need to compute $T(x_1,x_2,x_3)$. For this, set
		\[
		g(y) \ :=\ \left(1-\frac{|x_1-x_2|^{\frac{d-2s}{2}}}{|y-x_2|^{\frac{d-2s}{2}}}\right)\left(1-\frac{|x_1-x_3|^{\frac{d-2s}{2}}}{|y-x_3|^{\frac{d-2s}{2}}}\right),
		\]
		and
		\[
		f(y)\ := \ |y-x_1|^{-(d+2s-2)}.
		\]
		Due to identity~\eqref{lap-Riesz}, one has
		\[
		|y-x_1|^{-d-2s} \ =\ \frac{1}{2s(d+2s-2)}\,\Delta_yf(y),
		\]
		and therefore the integral can be rewritten as
		\begin{equation}\label{T}
			T(x_1,x_2,x_3) \ =\ \frac{1}{2s(d+2s-2)}\intd g(y)\,\Delta_yf(y)\,dy.
		\end{equation}
		Since Step 1 ensures that the above integral is finite, we may localise it on punctured domains and then pass to the limit. More precisely, 
		\[
		\intd g(y)\,\Delta_yf(y)\,dy \ =\ \lim_{\epsilon\downarrow0}\lim_{R\to\infty}\int_{B_{R,\epsilon}}g(y)\,\Delta_yf(y)\,dy,
		\]
		where we denote
		\[
		B_{R,\epsilon} \ :=\ B_R(0)\setminus (B_{\epsilon}(x_1)\cup B_{\epsilon}(x_2)\cup B_{\epsilon}(x_3)).
		\]
		Finally, applying Green's identity to the integral on $B_{R,\epsilon}$ gives
		\begingroup
		\begin{align}\label{eq:Green:full}
			\int_{B_{R,\epsilon}}g\,\Delta f \,dy & \ =\
			\int_{B_{R,\epsilon}}f\,\Delta g\,dy + \int_{\partial B_{R}(0)}(g\,\partial_{\nu}f-f\,\partial_{\nu}g)\,dS \\
			&\hspace{3em} - \sum^3_{i=1}\int_{\partial B_{\epsilon}(x_i)}(g\,\partial_{\nu}f-f\,\partial_{\nu}g)\,dS.\notag
		\end{align}
		\endgroup
		
		\smallskip
		\noindent\textbf{Step 3: Analysis of the boundary terms.}\\
		We next study the boundary contributions in the limit $R\to\infty$ and $\epsilon\downarrow 0$. For large $|y|$, one has
		\[
		g(y) \ =\ O(1)\quad\text{ and }\quad f(y) \ =\ O(|y|^{-d-2s+2}),
		\]
		together with
		\[
		|\nabla g(y)| \ =\ O\Bigl(|y|^{-\frac{d-2s}{2}-1}\Bigr)\quad\text{ and }\qquad |\nabla f(y)| \ =\ O\bigl(|y|^{-d-2s+1}\bigr).
		\]
		Hence, the dominant term in $g\,\partial_{\nu}f-f\,\partial_{\nu}g$ behaves like $O(|y|^{-d-2s+1})$, and therefore
		\[
		\int_{\partial B_R(0)}(g\,\partial_{\nu}f-f\,\partial_{\nu}g)\, dS\ =\ O(R^{-2s})\ =\ o(1)
		\quad \text{ as }\,R\to\infty.
		\]
		As $y\to x_1$, the cancellations in $g$ give
		\[
		g(y)\ =\ O\bigl(|y-x_1|^2\bigr),\qquad |\nabla g(y)|\ =\ O\bigl(|y-x_1|\bigr),
		\]
		whereas $f$ and $\nabla f$ have the same homogeneity as above with $|y|$ replaced by $|y-x_1|$.  It follows that
		\[
		g\,\partial_{\nu}f-f\,\partial_{\nu}g
		\ =\ O\Bigl(|y-x_1|^{-d-2s+3}\Bigr).
		\]
		Therefore, since $s<1$, one has
		\[
		\int_{\partial B_{\epsilon}(x_1)}(g\,\partial_{\nu}f-f\,\partial_{\nu}g)\,dS \ =\
		O\Bigl(\epsilon^{2-2s}\Bigr) \ =\ o(1)\quad\text{ as }\,\epsilon\downarrow 0.
		\]	
		We now study the boundary integral on $\partial B_{\epsilon}(x_2)$. The analysis on $\partial B_{\epsilon}(x_3)$ is analogous. Let us denote $m:=(d-2s)/2$ and let us write $y \in  \partial B_{\epsilon}(x_2)$ as $y=x_2+\epsilon\,\theta$ with $\theta\in\S^{d-1}$.
		Then,
		\[
		f(y) \ =\ |x_1-x_2|^{-(d+2s-2)}+O(\epsilon),\qquad \partial_{\nu}f(y)=O(1),
		\]
		and
		\[
		g(y) \ =\ \left(1-\frac{|x_1-x_2|^m}{\epsilon^m}\right)\left(1-\frac{|x_1-x_3|^m}{|x_2-x_3|^m}+O(\epsilon)\right) 
		\ =\ O(\epsilon^{-m}). 
		\]
		A simple calculation of $\nabla g$, followed by the evaluation of $\partial_{\nu}g(y)$ on $\partial B_{\epsilon}(x_2)$, shows that
		\[
		\partial_{\nu}g(y) \ =\ m\left(1-\frac{|x_1-x_3|^m}{|x_2-x_3|^m}\right)\,|x_1-x_2|^m\,\epsilon^{-m-1}+O(\epsilon^{-m}).    
		\]
		Since the surface area of $\partial B_{\epsilon}(x_2)$ is $|\S^{d-1}|\epsilon^{d-1}$, these estimates give
		\begin{eqnarray}\nonumber
			\lefteqn{\int_{\partial B_{\epsilon}(x_2)}\left(g\,\partial_{\nu}f-f\,\partial_{\nu}g\right)\, dS}\\
			&& \ =\ -|\S^{d-1}|\left(\frac{d-2s}{2}\right)\left(1-\frac{|x_1-x_3|^{\frac{d-2s}{2}}}{|x_2-x_3|^{\frac{d-2s}{2}}}\right)|x_1-x_2|^{2-\frac{d+6s}{2}}\,\epsilon^{\frac{d+2s-4}{2}} +O\Bigl(\epsilon^{\frac{d}{2}+s-1}\Bigr).\label{eq:Green:expansion:1} 
		\end{eqnarray}
		Arguing analogously near $x_3$ yields
		\begin{eqnarray}\nonumber
			\lefteqn{\int_{\partial B_{\epsilon}(x_3)}\left(g\,\partial_{\nu}f-f\,\partial_{\nu}g\right)\, dS}\\
			&& \ =\ -|\S^{d-1}|\left(\frac{d-2s}{2}\right)\left(1-\frac{|x_1-x_2|^{\frac{d-2s}{2}}}{|x_2-x_3|^{\frac{d-2s}{2}}}\right)|x_1-x_3|^{2-\frac{d+6s}{2}}\,\epsilon^{\frac{d+2s-4}{2}}+O\Bigl(\epsilon^{\frac{d}{2}+s-1}\Bigr). \label{eq:Green:expansion:2} 
		\end{eqnarray}
		In particular, both boundary terms vanish as $\epsilon\downarrow 0$, provided $d+2s>4$. In the borderline case, $d+2s=4$, they converge to finite non-zero limits and can be read off explicitly from~\eqref{eq:Green:expansion:1} and~\eqref{eq:Green:expansion:2}.\medskip
		
		\noindent\textbf{Step 4: Evaluation of the bulk term.}\\
		It remains to compute the bulk contribution $\int f\Delta g$. A direct calculation gives
		\begingroup
		\allowdisplaybreaks
		\begin{align*}\label{integrand-vol-term}
			f \Delta g &\ =\  \frac{(d-2s)(d+2s-4)}{4}\left[|x_1-x_2|^{\frac{d-2s}{2}} |y-x_2|^{-\frac{d-2s}{2}-2}|y-x_1|^{-d-2s+2} \right. \notag\\[0.25em]
			& \left.\;\; -|x_1-x_3|^{\frac{d-2s}{2}}|x_1-x_2|^{\frac{d-2s}{2}} |y-x_2|^{-\frac{d-2s}{2}-2}|y-x_3|^{-\frac{d-2s}{2}}|y-x_1|^{-d-2s+2}\right] \notag\\[0.25em]
			&+\frac{(d-2s)(d+2s-4)}{4}\left[|x_1-x_3|^{\frac{d-2s}{2}} |y-x_3|^{-\frac{d-2s}{2}-2}|y-x_1|^{-d-2s+2} \right. \notag \\[0.25em]
			&\left.\;\; -|x_1-x_2|^{\frac{d-2s}{2}} |x_1-x_3|^{\frac{d-2s}{2}} |y-x_3|^{-\frac{d-2s}{2}-2}|y-x_2|^{-\frac{d-2s}{2}}|y-x_1|^{-d-2s+2} \right] \notag \\[0.25em]
			&+2\left(\frac{d-2s}{2}\right)^2|x_1-x_3|^{\frac{d-2s}{2}}|x_1-x_2|^{\frac{d-2s}{2}}\frac{(y-x_2)}{|y-x_2|^{\frac{d-2s}{2}+2}}\cdot\frac{(y-x_3)}{|y-x_3|^{\frac{d-2s}{2}+2}}|y-x_1|^{-d-2s+2}.
		\end{align*}
		\endgroup
		We now evaluate the resulting integrals using:
		\begin{enumerate}[label=$\bullet$, itemsep=0.5em, leftmargin=*]
			\item the two-point beta integral~\eqref{beta-integral} with
			\[
			\alpha \ =\ \frac{d-2s}{2}+2,\qquad \beta \ =\ d+2s-2
			\]
			\item the three-fold identity~\eqref{eq:prelim:grafakos:formula} with
			\[
			\alpha_2 \ =\ \frac{d-2s}{2}+2,\qquad \alpha_3\ =\frac{d-2s}{2},\qquad \alpha_1 \ =\ d+2s-2
			\]
			\item and Theorem~\ref{thm:main:2} with
			\[
			d_2 \ =\ \frac{d-2s}{2}+1,\qquad d_3 \ =\ \frac{d-2s}{2}+1,\qquad d_1\ =\ d+2s-2\,,
			\]
		\end{enumerate}
		to obtain
		\begingroup
		\allowdisplaybreaks
		\begin{align}
			\intd f \Delta g\, dy 
			&\ = \ \frac{(d-2s)(d+2s-4)}{4} \left[|x_1-x_2|^{\frac{d-2s}{2}} \mathcal{B}(x_2,x_1)+|x_1-x_3|^{\frac{d-2s}{2}} \mathcal{B}(x_3,x_1)\right. \nonumber \\
			&\hspace{3em}\left. -|x_1-x_2|^{\frac{d-2s}{2}}|x_1-x_3|^{\frac{d-2s}{2}}\left(\GM(x_2,x_3,x_1) + \GM(x_3,x_2,x_1)\right) \right]  \nonumber \\
			&\hspace{3em}+\frac{(d-2s)^2}{2}|x_1-x_2|^{\frac{d-2s}{2}}|x_1-x_3|^{\frac{d-2s}{2}}\, K(x_2,x_3,x_1). \label{bulk-integral} 
		\end{align}
		\endgroup
		\medskip
		
		\noindent\textbf{Step 5: Evaluation of the right-hand side of~\eqref{interaction-term}}.\\
		In the regime $d+2s>4$, the boundary terms vanish and only the bulk term contributes. 		
		By making use of the following calculations
		\begingroup
		\allowdisplaybreaks
		\begin{align*}
			&(a.1) \qquad |x_i-x_j|^{\frac{d-2s}{2}} |x_i-x_j|^{d-\alpha-\beta} \ =\   |x_i-x_j|^{-2s}\,.\\
			&(a.2) \qquad|x_1-x_2|^{\frac{d-2s}{2}}|x_1-x_3|^{\frac{d-2s}{2}} |x_2-x_3|^{\alpha_1-d} |x_3-x_1|^{\alpha_2-d}|x_1-x_2|^{\alpha_3-d} \\
			&\hspace{16em} \ =\ \dfrac{|x_1-x_3|^{2-2s}}{|x_1-x_2|^{2s}|x_2-x_3|^{2-2s}}\,.\\[0.5em]
			&(a.3) \qquad|x_1-x_3|^{\frac{d-2s}{2}}|x_1-x_2|^{\frac{d-2s}{2}} |x_3-x_2|^{\alpha_1-d} |x_2-x_1|^{\alpha_2-d}|x_3-x_1|^{\alpha_3-d}\\ 
			&\hspace{16em}\ =\   \dfrac{|x_1-x_2|^{2-2s}}{|x_1-x_3|^{2s}|x_2-x_3|^{2-2s}}\,.\\[0.5em]
			&(a.4) \qquad |x_1-x_2|^{\frac{d-2s}{2}}|x_1-x_3|^{\frac{d-2s}{2}} \frac{(x_2-x_1)}{|x_2-x_1|^{d-d_3+1}}\cdot \frac{(x_3-x_1)}{|x_3-x_1|^{d-d_2+1}}\frac{1}{|x_2-x_3|^{d-d_1}}\\
			& \hspace{16em} \ =\  \frac{(x_2-x_1)}{|x_2-x_1|^{2s}}\cdot \frac{(x_3-x_1)}{|x_3-x_1|^{2s}}\frac{1}{|x_2-x_3|^{2-2s}}\,.
		\end{align*}	
		\endgroup	
		and
		\begingroup
		\allowdisplaybreaks
		\begin{align*}
			& (b.1) \qquad  \dfrac{2 a_{s,d}}{2s(d+2s-2)} \dfrac{(d-2s)(d+2s-4)}{4}\, C_{\rm Stein}(\alpha,\beta,d) \ =\  \CFH\,.\\[0.5em]
			& (b.2) \qquad \dfrac{2 a_{s,d}}{2s(d+2s-2)} \dfrac{(d-2s)(d+2s-4)}{4}\, C_{\rm GM}(\alpha_1,\alpha_2,\alpha_3,d) \ =\  \CFH\,.\\[0.5em]
			& (b.3) \qquad  \dfrac{2 a_{s,d}}{2s(d+2s-2)} \dfrac{(d-2s)^2}{2}\, C_{\rm MOZ}(d_1,d_2,d_3,d)\ =\  2\,\CFH\,.
		\end{align*}	
		\endgroup	
		we conclude that
		\begingroup
		\begin{align*}        
			&\lim_{\alpha\uparrow\frac{d-2s}{2}}\lim_{n\to\infty}\frac{\Nn(\valphan(\,\cdot-x_2),\valphan(\,\cdot-x_3))(x_1)}{\valphan(x_1-x_2)\,\valphan(x_1-x_3)}\\[0.25em]
			&\ =\ \CFH\,\left[\frac{1}{|x_1-x_2|^{2s}}+\frac{1}{|x_1-x_3|^{2s}}+2\,\frac{(x_3-x_1)}{|x_3-x_1|^{2s}}\cdot\frac{(x_2-x_1)}{|x_2-x_1|^{2s}}\frac{1}{|x_2-x_3|^{2-2s}}\right.\\[0.25em]
			&\left. \hspace{5em}-\left(\frac{|x_1-x_3|^{2-2s}}{|x_1-x_2|^{2s}|x_2-x_3|^{2-2s}}+\frac{|x_1-x_2|^{2-2s}}{|x_1-x_3|^{2s}|x_2-x_3|^{2-2s}}\right)\right].
		\end{align*}
		\endgroup
		It remains to consider the borderline case $d+2s=4$. Since $d\in\N$ and $0<s<1$, this necessarily means 
		\[
		(d,s) \ =\ \Bigl(3,\frac{1}{2}\Bigr).
		\]
		In particular,
		\[
		\frac{d-2s}{2} \ =\ 1,\qquad 2s(d+2s-2) \ =\ 2,\qquad |\S^{d-1}|\ =\ |\S^2|\ =\ 4\pi, \qquad  a_{\frac{1}{2},3} \ =\ \frac{1}{2\pi^2}\,.
		\]
		Moreover, the factor $d+2s-4$ vanishes, and hence the bulk contribution reduces to the last term in~\eqref{bulk-integral} which, in view of $(b.3)$, is nothing but
		\[
		2 C_{fH}\Bigl(3,\frac{1}{2}\Bigr)\;\frac{(x_3-x_1)}{|x_3-x_1|}\cdot\frac{(x_2-x_1)}{|x_2-x_1|}\,\frac{1}{|x_2-x_3|}.
		\]
		On the other hand, for this special case $(d,s)=(3,\tfrac{1}{2})$, we calculate the boundary expansions~\eqref{eq:Green:expansion:1} and~\eqref{eq:Green:expansion:2} multiplied by the overall prefactor $\dfrac{2 a_{s,d}}{2s(d+2s-2)}= \dfrac{1}{2 \pi^2}$, to obtain
		\begingroup
		\allowdisplaybreaks
		\begin{align*}       
			& -\frac{1}{2 \pi^2}\lim_{\epsilon\downarrow 0}\left(
			\int_{\partial B_{\epsilon}(x_2)}(g\,\partial_{\nu}f-f\,\partial_{\nu}g)\,dS+\int_{\partial B_{\epsilon}(x_3)}(g\,\partial_{\nu}f-f\,\partial_{\nu}g)\,dS\right)\\
			& \ =\ \dfrac{2}{\pi}\left(\frac{1}{|x_1-x_2|}+\frac{1}{|x_1-x_3|}-\frac{|x_1-x_3|}{|x_1-x_2||x_2-x_3|}-\frac{|x_1-x_2|}{|x_1-x_3||x_2-x_3|}\right).
		\end{align*}
		\endgroup
		Now, observing that $C_{fH}\Bigl(3,\frac{1}{2}\Bigr) \ =\ \dfrac{2}{\pi}$, it follows that the total boundary contribution is
		\[
		C_{fH}\Bigl(3,\frac{1}{2}\Bigr) \left(\frac{1}{|x_1-x_2|}+\frac{1}{|x_1-x_3|}-\frac{|x_1-x_3|}{|x_1-x_2||x_2-x_3|}-\frac{|x_1-x_2|}{|x_1-x_3||x_2-x_3|}\right).
		\]
		Adding this term to the angular bulk term yields precisely
		\begingroup
		\allowdisplaybreaks
		\begin{align*}
			&C_{fH}\Bigl(3,\frac{1}{2}\Bigr) \left(\frac{1}{|x_1-x_2|}+\frac{1}{|x_1-x_3|}+2\,\frac{(x_3-x_1)}{|x_3-x_1|}\cdot\frac{(x_2-x_1)}{|x_2-x_1|}\,\frac{1}{|x_2-x_3|}\right.\\[0.25em]
			&\hspace{5em}-\left.\left[\frac{|x_1-x_3|}{|x_1-x_2||x_2-x_3|}+\frac{|x_1-x_2|}{|x_1-x_3||x_2-x_3|}\right]\right).
		\end{align*}
		\endgroup
		This is exactly the specialisation of~\eqref{cross-term-limit} to $s=1/2$. Therefore, the same closed formula holds in the borderline case as well. The proof of the lemma is now complete.
	\end{proof}
	\smallskip
	
	\subsection{Proof of Claim~\eqref{FT-product} in Theorem~\ref{thm:main:2} }\label{subsec:app:A:2}
	
	We now prove Claim~\eqref{FT-product} used in the derivation of Theorem~\ref{thm:main:2}. The issue is to justify a Fourier transform formula for a product of tempered distributions where the standard formulae do not apply directly. We therefore argue by truncation and passage to the limit. 
	\begin{proposition}\label{prop:app:Fourier:convolution}
		Let $0<d_1,d_3<d$ be exponents with $d_1+d_3>d$. Then the following convolution formula holds in the sense of tempered distributions:
		\begin{equation}\label{eq:lem:app:Fourier:convolution}
			\left(\frac{w_1}{|w|^{d-d_3+1}|w-e_1|^{d-d_1}}\right)\widehat{\hspace{9pt}}(\xi)
			\, =\,  \frac{\CLL{d_1}\CLL{d_3+1}}{i\,\CLL{d-d_1}\CLL{d-d_3+1}} \left( \frac{y_1}{|y|^{d_3+1}} \ast \frac{e^{-2 \pi i y \cdot e_1}}{|y|^{d_1}} \right)(\xi)\,\text{ in }\calS'(\R^d)
		\end{equation}
		where $e_1$ is the first standard basis vector of $\R^d$, and $\CLL{\mu}:=\pi^{-\frac{\mu}{2}}\Gamma(\frac{\mu}{2})$ for $\mu>0$.
	\end{proposition}
	\begin{proof}[Proof of Proposition~\ref{prop:app:Fourier:convolution}] We already have that the product of the tempered distributions $w_1/|w|^{d-d_3+1}$ and $\tau^{e_1}(1/|w|^{d-d_1})$ is precisely the function $J$ introduced in~\eqref{Jw} (see Step 2 in the proof of Theorem~\ref{thm:main:2}). Although $J$ does not belong to $L^1(\R^d)$, it defines a tempered distribution and therefore admits a Fourier transform. 
		However, the standard sufficiency conditions for the Fourier transform of a  product, such as formula \eqref {FTprop-distributions} in Grafakos~\cite[Proposition 2.3.22]{grafakos2008classical} or in Tr\`{e}ves~\cite[Theorem 30.4]{Treves}, do not hold. Hence, some truncation and regularisation arguments are needed before such a formula can be applied, and the desired identity will be obtained via a limit argument. This will be carried out in four steps, as detailed below. 
		
		\medskip
		\noindent
		Let \(\chi\in C_c^\infty(\R^d;[0,1])\) be radial, with \(\chi\equiv1\) on \(\{|x|\le1\}\) and \(\supp \chi\subset\{|x|\le2\}\). For \(L>1\), define
		\[
		K_L(x)\ :=\ \frac{x_1}{|x|^{d_3+1}}\chi\left(\frac{x}{L}\right),
		\qquad
		\psi(w)\ :=\ \frac{\widehat{\varphi}(w)}{|w-e_1|^{d-d_1}},
		\]
		where \(\varphi\in\calS(\R^d)\) is fixed.\smallskip 
			
				\noindent\textbf{Step 1:} 
		Since both \(K_L\) and $\psi$ are integrable on \(\R^d\), Fubini's theorem yields
		\begin{equation}    \label{basic-identity}    
			\intd \widehat{K_L}(w)\psi(w)\,dw
			\ =\ 
			\intd K_L(x)\widehat{\psi}(x)\,dx.
		\end{equation}
		The desired result will follow by passing to the limit as $L \to \infty$.\smallskip
		
		\noindent\textbf{Step 2:}  In view of the last property in~\eqref{FTprop-distributions}, the second property in \eqref{FTprop-functions}, and the third property in \eqref{FTprop-distributions}, one has	
		\begingroup
		\allowdisplaybreaks
		\begin{align}\label{eq:app:proof:Fourier:identity:aux:2}          
			\widehat{\psi} \ =\	\left(\frac{1}{|\cdot-e_1|^{d-d_1}}\widehat{\varphi}\right)\widehat{\hspace{9pt}}
			&\ =\  \left(\frac{1}{|\cdot-e_1|^{d-d_1}}\right)\widehat{\hspace{9pt}}\ast\widetilde{\varphi}(\cdot\,)\hspace{3.8em}\text{ in }\calS'(\R^d)\notag\\
			&\ =\  e^{-2\pi i(\cdot)\cdot e_1}\left(\frac{1}{|\cdot|^{d-d_1}}\right)\widehat{\hspace{9pt}}\ast\widetilde{\varphi}(\,\cdot\,)\notag\\
			&\ =\ \frac{\CLL{d_1}}{\CLL{d-d_1}} \Bigl(\frac{e^{-2\pi i(\cdot)\cdot e_1}}{|\cdot|^{d_1}}\Bigr)\ast \widetilde{\varphi}(\,\cdot\,)\hspace{2em}\text{ in }\calS'(\R^d)\,,
		\end{align}
		\endgroup
		the last step being a consequence of the Fourier transform formula~\eqref{eq:Fourier:Lieb-Loss:aux:1}. Since \(\psi\in L^1(\R^d)\), its Fourier transform is bounded and continuous, whereas the right-hand side of \eqref{eq:app:proof:Fourier:identity:aux:2} is smooth. {Therefore,  \eqref{eq:app:proof:Fourier:identity:aux:2} also holds pointwise.} 
		
		Applying the same product--convolution rule to the truncated kernel gives
		\begin{equation}\label{eq:app:proof:Fourier:identity:aux:2:right}        
			\left(
			\frac{x_1}{|x|^{d_3+1}}\chi\left(\frac{x}{L}\right)
			\right)\widehat{\hspace{9pt}}
			=
			\widehat{\frac{x_1}{|x|^{d_3+1}}}
			\ast
			\widehat{\chi\Bigl(\frac{\cdot}{L}\Bigr)}
			\qquad\text{in }\calS'(\R^d).
		\end{equation}
		As before, this identity holds pointwise.\smallskip
		
		\noindent\textbf{Step 3:} Using \eqref{eq:app:proof:Fourier:identity:aux:2} from Step 2 and Lemma~\ref{lem:app:nested:conv:estimate:2}, we see that the integrand on the right-hand side of \eqref{basic-identity} can be estimated (aside from the constant $\CLL{d_1}/\CLL{d-d_1}$) as follows:
		\begin{eqnarray*}        
			\left|\frac{x_1}{|x|^{d_3+1}}\chi\left(\frac{x}{L}\right)\left(\frac{e^{-2\pi i(\cdot)\cdot e_1}}{|\cdot|^{d_1}}\ast \widetilde{\varphi} \right)(x)\right|
			&\ \lesssim & \ C_{d,d_1,\varphi}|x|^{-d_3}(1+|x|)^{-d_1}\\
			&\ \lesssim &\ C_{d,d_1,\varphi}(|x|^{-d_3}\1_{\{|x|\leq 1\}}+|x|^{-(d_1+d_3)}\1_{\{|x|> 1\}}).
		\end{eqnarray*}
		Thus, we obtain domination by an $L^1(\R^d)$-function that is independent of $L$, guaranteed by the conditions $0<d_3<d$ and $d_1+d_3>d$. Therefore, Lebesgue's dominated convergence theorem yields
		\begin{eqnarray}
			\lim_{L \to \infty} \intd K_L(x)\widehat{\psi}(x)\,dx 
			& = &  \frac{\CLL{d_1}}{\CLL{d-d_1}} \intd\frac{x_1}{|x|^{d_3+1}} \left(\frac{e^{-2\pi i(\cdot)\cdot e_1}}{|\cdot|^{d_1}}\ast \widetilde{\varphi} \right)(x)\,dx \nonumber \\
			& = &  -\frac{\CLL{d_1}}{\CLL{d-d_1}} \intd\left(\frac{(\cdot)\cdot e_1}{|\cdot|^{d_3+1}} \ast \frac{e^{-2\pi i(\cdot)\cdot e_1}}{|\cdot|^{d_1}}\right)(x)\,\varphi(x)\,dx\,.	\label{eq:app:proof:Fourier:identity:aux:last}
		\end{eqnarray}   
		The last identity follows from an application of Fubini's theorem, which is legitimate since the uniform integrability of $K_L\widehat{\psi}$ implies the integrability of the right-hand side.\smallskip

		\noindent\textbf{Step 4:} We now pass to the limit in the left-hand side of \eqref{basic-identity}. To begin with, as a consequence of \eqref{eq:app:proof:Fourier:identity:aux:2:right}, there holds
		\[
		\intd \widehat{K_L}(w)\psi(w)\,dw   
		\ =\ \intd\left(\widehat{\frac{x_1}{|x|^{d_3+1}}}\ast\widehat{\chi\Bigl(\frac{\cdot}{L}\Bigr)}\right)(w) \frac{\widehat{\varphi}(w)}{|w-e_1|^{d-d_1}}\,dw\,.    
		\]     
		Then, in view of~\eqref{lem:app:Fourier},  $K:=(x_1/|x|^{d_3+1})\widehat{\hspace{9pt}}=\Bigl(\frac{\CLL{d-d_3+1}}{i\,\CLL{d_3+1}}\Bigr)\frac{w_1}{|w|^{d-d_3+1}}$ and so belongs to $L^1_{loc}(\R^d)$. Therefore, invoking Lemma~\ref{lem:app:nested:conv:estimate:3}, with $K=(x_1/|x|^{d_3+1})\widehat{\hspace{9pt}}$,  $\psi=\widehat{\chi(\cdot/L)}$, $\Phi=\widehat{\,\varphi\,}/|\cdot-e_1|^{d-d_1}$, 
		\begingroup
		\allowdisplaybreaks
		\begin{align*}   
			\intd\left(\widehat{\frac{x_1}{|x|^{d_3+1}}}\ast  \widehat{\chi\Bigl(\frac{\cdot}{L}\Bigr)}\right)(w)
			& \frac{\widehat{\varphi}(w)}{|w-e_1|^{d-d_1}}\,dw\\   
			& \ =\ \intd\widehat{\frac{x_1}{|x|^{d_3+1}}}\left( \frac{\widehat{\varphi}}{|\cdot-e_1|^{d-d_1}} \ast \widehat{\chi\Bigl(\frac{\cdot}{L}\Bigr)} \right)(w)\,dw\\ 
			& \ =\  \frac{\CLL{d-d_3+1}}{i\,\CLL{d_3+1}} \intd \frac{w_1}{|w|^{d-d_3+1}}
			\left(\frac{\widehat{\varphi}}{|\cdot-e_1|^{d-d_1}} \ast \widehat{\chi\Bigl(\frac{\cdot}{L}\Bigr)}
			\right)(w)\,dw.
		\end{align*}
		\endgroup
		We claim that		\begin{equation}\label{eq:app:nested:conv:estimate:3:bis:aux:last:2}
			\lim_{L\to\infty} \intd \frac{w_1}{|w|^{d-d_3+1}} \left( \frac{\widehat{\varphi}}{|\cdot-e_1|^{d-d_1}} \ast \widehat{\chi\Bigl(\frac{\cdot}{L}\Bigr)}     \right)(w)\,dw 
			\ =\  \intd\frac{w_1}{|w|^{d-d_3+1}}\frac{\widehat{\varphi}}{|w-e_1|^{d-d_1}}\,dw.
		\end{equation}
		To see this, observe that 
		\[
		\Phi(w)\ :=\ \frac{\widehat{\varphi}(w)}{|w-e_1|^{d-d_1}}\in L^p(\R^d)\;\;\text{ for }\;\; 1\leq p<\frac{d}{d-d_1}.
		\]    
		Thus, by the standard \(L^p\)-convergence theorem for convolutions, \begin{equation}\label{eq:app:nested:conv:estimate:3:bis:aux:1}
			\left\|\Phi\ast\widehat{\chi\Bigl(\frac{\cdot}{L}\Bigr)}-\Phi\right\|_{L^p(\R^d)}\xrightarrow[L\to\infty]{}0.        
		\end{equation}
		\smallskip
		
		\noindent Let us now introduce the convergence error $E_L$ as follows
		\begin{equation*}
			E_L 
			=\ \intd K(w)\biggl(\Phi\ast\widehat{\chi\Bigl(\frac{\cdot}{L}\Bigr)}-\Phi\biggr)(w)\,dw\,.
		\end{equation*}
		We split the error into two parts
		\begin{align*}
			E_L \ &=\ \int_{\{|w|\leq R\}} K(w)\biggl(\Phi\ast\widehat{\chi\Bigl(\frac{\cdot}{L}\Bigr)}-\Phi\biggr)(w)\,dw +\int_{\{|w|>R\}} K(w)\biggl(\Phi\ast\widehat{\chi\Bigl(\frac{\cdot}{L}\Bigr)}-\Phi\biggr)(w)\,dw \\ 
			&:=\ E_{{\rm inn},R}(L)+E_{{\rm out},R}(L)\,.
		\end{align*}
		Since $0<d_3<d$ and $d_1+d_3>d$, the interval $\left(\frac{d}{d_3},\frac{d}{d-d_1}\right)$ is non-empty. Thus, if we choose $p \in \left(\frac{d}{d_3},\frac{d}{d-d_1}\right)$, then its H\"older conjugate $p'=p/(p-1)$ satisfies $1<p'<d/(d-d_3)$. In particular, $K\in L^{p'}(\{|w|\leq R\})$ for every $R>0$. Hence, H\"older's inequality and~\eqref{eq:app:nested:conv:estimate:3:bis:aux:1} imply
		\[
		|E_{{\rm inn}, R}(L)|
		\ \leq\ 
		\|K\|_{L^{p'}(\{|w|\leq R\})}
		\left\|\Phi\ast\widehat{\chi\Bigl(\frac{\cdot}{L}\Bigr)}-\Phi\right\|_{L^p(\R^d)}
		\xrightarrow[L\to\infty]{}0,
		\]
		for each fixed $R>0$. 
		
		\smallskip
		\noindent Now, note that $\|\widehat{\chi(\cdot/L)}\|_{L^1(\R^d)}=\|\widehat{\chi}\|_{L^1(\R^d)}$ is independent of \(L\). So, by Young's inequality,
		\[
		\left\|\Phi\ast\widehat{\chi\Bigl(\frac{\cdot}{L}\Bigr)}-\Phi\right\|_{L^1(\R^d)}
		\ \leq\  \bigl(1+\|\widehat{\chi}\|_{L^1(\R^d)}\bigr)\|\Phi\|_{L^1(\R^d)}
		=:C_0,
		\]
		where \(C_0\) is independent of $L$ for \(L>1\). Therefore,
		\begin{align*}
			|E_{{\rm out},R}(L)|
			&\ \leq\ \biggl(\sup_{|w|>R}|K(w)|\biggr)\,
			\left\|\Phi\ast\widehat{\chi\Bigl(\frac{\cdot}{L}\Bigr)}-\Phi\right\|_{L^1(\R^d)}\\
			&\ \leq \ C_0\sup_{|w|>R}|K(w)|\ \lesssim \frac{C_0}{R^{d-d_3}} 
		\end{align*}
		and thus $\lim_{R\to+\infty}\sup_{L>1}|E_{{\rm out},R}(L)|=0$. Combining this with the observed fact that $E_{{\rm inn}, R}(L)\to 0$ for any fixed $R>>1$ yields $E_L\to 0$ as $L\to\infty$. This proves the claim~\eqref{eq:app:nested:conv:estimate:3:bis:aux:last:2}. 
		
		\smallskip
		\noindent
		The desired result follows by passing to the limit in  \eqref{basic-identity}, using \eqref{eq:app:nested:conv:estimate:3:bis:aux:last:2} and \eqref{eq:app:proof:Fourier:identity:aux:last}.
	\end{proof}
	
	\subsection{Convolution properties}\label{subsec:app:A:3}
	
	\begin{lemma}\label{lem:app:nested:conv:estimate:3}
		Let $K\in L^1_{loc}(\R^d)$ be bounded at infinity, and $\Phi\in L^1(\R^d)$. Suppose that $\psi\in L^1(\R^d)\cap L^{\infty}(\R^d)$. Then the following convolution formula holds:
		\[
		\intd(K\ast \psi)(w)\Phi(w)\,dw \ =\ \intd K(w)(\Phi\ast\widetilde{\psi})(w)\,dw. 
		\]
	\end{lemma}
	\begin{proof}[Proof of Lemma~\ref{lem:app:nested:conv:estimate:3}]
		It suffices to establish the absolute integrability:
		\begin{equation}\label{eq:app:nested:conv:estimate:3:aux:1}            
			\iint_{\R^d\times\R^d}|K(z)\psi(w-z)\Phi(w)|\,dz\,dw<\infty.
		\end{equation}
		For $w\in\R^d$, define $\Theta(w):=\intd|K(u)||\psi(w-u)|\,du$. Then the above can be rewritten as
		\[
		\iint_{\R^d\times\R^d}|K(z)\psi(w-z)\Phi(w)|\,dz\,dw \ =\ \intd|\Phi(w)|\Theta(w)\,dw.
		\]
		As $\Phi\in L^1(\R^d)$, to justify~\eqref{eq:app:nested:conv:estimate:3:aux:1} it is enough to prove that $\Theta\in L^{\infty}(\R^d)$. To see this, note that the assumptions on $K$ ensure that there exists $R_0>1$ such that $K\in L^1(\{|z|\leq R\})$ and $K\in L^{\infty}(\{|z|>R\})$ for every $R\geq R_0$. Combining this with $\psi\in L^1(\R^d)\cap L^{\infty}(\R^d)$ yields
		\[
		\begin{aligned}
			\Theta(w)
			&\ =\ \int_{\{|u|\leq R_0\}}|K(u)||\psi(w-u)|\,du+ \int_{\{|u|>R_0\}}|K(u)||\psi(w-u)|\,du\\
			&\ \leq\ \|K(z)\|_{L^1(\{|z|\leq R_0\})}\|\psi\|_{L^{\infty}(\R^d)}+\|K(z)\|_{L^{\infty}(\{|z|>R_0\})}\|\psi\|_{L^1(\R^d)}<\infty.
		\end{aligned}
		\]
		As this bound is independent of $w\in\R^d$, $\Theta$ is bounded on $\R^d$. Then, by Fubini's theorem,
		\begingroup
		\allowdisplaybreaks
		\begin{align*}
			\intd(K\ast\psi)(w)\Phi(w)\,dw
			&\ =\ \intd K(z)\left(\intd\psi(w-z)\Phi(w)\,dw\right)dz\\
			&\ =\ \intd K(z)(\psi(-\,\cdot\,)\ast\Phi)(z)\,dz,
		\end{align*}
		\endgroup
		and the desired identity follows.
	\end{proof}
	
	\begin{lemma}\label{lem:app:nested:conv:estimate:1}
		Suppose that $\eta\in L^1_{\alpha}(\R^d)$ and $\zeta\in L^{\infty}_{\alpha}(\R^d)$ for some $\alpha>0$; that is, 
		\[
		\intd|\eta(y)|(1+|y|)^{\alpha}\,dy<\infty\,,
		\]
		and 
		\[
		\esssup_{y\in\R^d} (1+|y|)^{\alpha}|\zeta(y)|
		\ \leq\ C_{\alpha,\zeta}<\infty.
		\]    
		Then,  
		\[
		\eta\ast\zeta\in L^{\infty}_{\alpha}(\R^d),
		\]
		and moreover
		\[
		\esssup_{\R^d}(1+|\cdot|)^{\alpha}|\eta\ast\zeta|\ \leq\ C_{\alpha,\zeta}\|(1+|\cdot|)^{\alpha}\eta\|_{L^1(\R^d)}.
		\]
	\end{lemma}
	\begin{proof}[Proof of Lemma~\ref{lem:app:nested:conv:estimate:1}]
		The claim is a weighted Young-type estimate. The key observation is that the weight $(1+|x|)^{\alpha}$ may be transferred from the output variable $x$ to the integration variable $y$ by means of the elementary inequality	
		\[
		(1+|x-y|)^{-\alpha}\ \leq\
		(1+|y|)^{\alpha} (1+|x|)^{-\alpha}\quad\text{ for }x,y\in\R^d,
		\]
		in view of the monotonicity of $t\mapsto t^{-\alpha}$ and  $1+|x|\leq (1+|y|)(1+|x-y|)$ for any $x,y\in\R^d$. 
		The $L^{\infty}$-bound on $\zeta$ allows one to estimate the convolution:
		\[
		|(\eta\ast\zeta)(x)| \ \leq\ \intd|\eta(y)|\,|\zeta(x-y)|\,dy
		\ \leq \ C_{\alpha,\zeta}\intd|\eta(y)|(1+|x-y|)^{-\alpha}\,dy.
		\]
		Applying the previous weight inequality yields
		\begin{equation*}
			|(\eta\ast\zeta)(x)|
			\ \leq\ C_{\alpha,\zeta}\, (1+|x|)^{-\alpha}\intd |\eta(y)|(1+|y|)^{\alpha}\,dy.		      
		\end{equation*}    
		Equivalently, we have just proved that
		\[
		\esssup_{x\in\R^d} (1+|x|)^{\alpha}|(\eta\ast\zeta)(x)| \ \leq\ C_{\alpha,\zeta} \|(1+|\cdot|)^{\alpha}\eta\|_{L^1(\R^d)}.
		\]
		This proves the lemma.
	\end{proof}
	\begin{lemma}\label{lem:app:nested:conv:estimate:2}
		Let $0<\alpha_1<d$ and $\varphi\in\calS(\R^d)$. Then the following decay holds
		\[
		\left|\left(e^{2\pi i(\cdot)\cdot e_1}|\cdot|^{-\alpha_1}\ast \varphi\right)(x)\right|\ \leq\  C_{d,\alpha_1,\varphi}\,(1+|x|)^{-\alpha_1}\quad\text{ for all }\,x\in\R^d,
		\]
		for some constant $C_{d,\alpha_1,\varphi}>0$.
	\end{lemma}
	
	\begin{proof}[Proof of Lemma~\ref{lem:app:nested:conv:estimate:2}]
		We split the kernel into its local and nonlocal parts. The local part is integrable and can be treated directly using the rapid decay of $\varphi$, whereas the nonlocal part is not integrable near infinity but satisfies a weighted $L^{\infty}$ bound, so Lemma~\ref{lem:app:nested:conv:estimate:1} applies.
		
		Set the function 
		\[
		\zeta:=e^{2\pi i (\cdot)\cdot e_1}|\cdot|^{-\alpha_1}\ast \varphi,
		\]    
		and decompose it as
		\[
		\begin{aligned}
			\zeta\ =\ \zeta_{0}\ +\ \zeta_{\infty}
			\ :=\ a_0\ast\varphi\ +\ a_{\infty}\ast\varphi,
		\end{aligned}
		\]
		where 
		\[
		a_0(y):=(e^{2\pi iy\cdot e_1}|y|^{-\alpha_1})\1_{\{|y|\leq 1\}},\qquad a_{\infty}(y):=(e^{2\pi iy\cdot e_1}|y|^{-\alpha_1})\1_{\{|y|>1\}}.
		\]    
		We first estimate the local part, $\zeta_0$. Since $0<\alpha_1<d$, the kernel $a_0$ belongs to $L^1(\R^d)$. Hence,
		\[
		|\zeta_0(x)|\ \leq\ \int_{\{|y|\leq 1\}}|y|^{-\alpha_1}|\varphi(x-y)|\,dy
		\ \leq\ \left(\int_{\{|y|\leq 1\}}|y|^{-\alpha_1}\,dy\right)\sup_{|u|\leq 1}|\varphi(x-u)|.
		\]
		Because $\varphi$ is a Schwartz function, there exists $\varsigma_{d,\varphi}>0$ such that
		\[
		|\varphi(y)|\leq \varsigma_{d,\varphi}\,(1+|y|)^{-d}
		\quad\text{ for all }y\in\R^d. 
		\]
		Moreover, if $|u|\leq 1$ then $1+|x-u| \geq \tfrac{1}{2}(1+|x|)$ for any $x\in\R^d$. Therefore, as $\alpha_1<d$, 
		\[
		\begin{aligned}        
			|\zeta_0(x)|
			\ \leq\ \varsigma_{d,\varphi}\||y|^{-\alpha_1}\|_{L^1(\{|y|\leq 1\})}\sup_{|u|\leq 1}(1+|x-u|)^{-d}
			\ \leq\ C^{(0)}_{\alpha_1,d,\varphi}(1+|x|)^{-\alpha_1}\,,
		\end{aligned}
		\]
		for every $x\in\R^d$ with constant $C^{(0)}_{\alpha_1,d,\varphi}:=2^d\varsigma_{d,\varphi}\||y|^{-\alpha_1}\|_{L^1(\{|y|\leq 1\})}\,$.	
		
		We now study the nonlocal part, $\zeta_{\infty}$. For $|y|>1$,
		\[
		|a_{\infty}(y)| \ =\ |y|^{-\alpha_1}\ \leq \ 2^{\alpha_1}(1+|y|)^{-\alpha_1},
		\]
		so $a_{\infty}\in L^{\infty}_{\alpha_1}(\R^d)$. On the other hand, since $\varphi\in\calS(\R^d)$, one has
		\[
		\intd|\varphi(y)|(1+|y|)^{\alpha_1}\,dy<\infty,   
		\]
		that is, $\varphi\in L^1_{\alpha_1}(\R^d)$. Applying Lemma~\ref{lem:app:nested:conv:estimate:1} in the case $\alpha=\alpha_1>0$, $\eta=\varphi\in L^1_{\alpha}(\R^d)$, $\zeta=a_{\infty}\in L^{\infty}_{\alpha}(\R^d)$ gives
		\[
		a_{\infty}\ast\varphi\in L^{\infty}_{\alpha_1}(\R^d).
		\]
		In particular,
		\[
		|\zeta_{\infty}(x)| \ \leq\ C^{(\infty)}_{d,\alpha_1,\varphi}\, (1+|x|)^{-\alpha_1}.
		\]
		Combining the two estimates, we conclude that
		\[
		|\zeta(x)|\ \leq\ |\zeta_0(x)|+|\zeta_{\infty}(x)|\ \leq\ C_{d,\alpha_1,\varphi}(1+|x|)^{-\alpha_1} \quad \text{ for all }x\in\R^d,
		\]
		with explicit constant $C_{d,\alpha_1,\varphi}:= C^{(0)}_{\alpha_1,d,\varphi}+C^{(\infty)}_{\alpha_1,d,\varphi}$. This proves the claim.
	\end{proof}

	\bigskip
	
	\bibliographystyle{amsplain}
	\bibliography{citations}
	
	\bigskip
	
	\begingroup
	\small
	\addressblock{Departamento de Matem\'atica, Universidad de Concepci\'on (UdeC), Avenida Esteban Iturra  s/n Barrio Universitario, Concepci\'on, Chile.}
	
	\EmailB{rmahadevan@udec.cl}
	\bigskip
	
	\addressblock{Instituto de Ciencias de la Ingenier\'ia, Universidad de O'Higgins (UOH), Avenida del Libertador Bernardo O'Higgins 611, Rancagua, Chile.}
	\EmailB{franco.olivares@uoh.cl}
	
	\EmailB{andres.zuniga@uoh.cl}
	\endgroup
	
\end{document}